\documentclass[11pt]{article}

\usepackage[utf8x]{inputenc}
\usepackage[english]{babel}

\usepackage{algorithmic}
\usepackage{algorithm}

\usepackage{mathtools}

\usepackage{nicefrac}

\usepackage{xcolor}

\usepackage[neverdecrease]{paralist}

\usepackage{amsmath, amsthm, amssymb}

\theoremstyle{plain}
\newtheorem{thm}{Theorem}[section]
\newtheorem{lem}[thm]{Lemma}
\newtheorem{prop}[thm]{Proposition}
\newtheorem{cor}[thm]{Corollary}
\theoremstyle{definition}
\newtheorem{defn}[thm]{Definition}

\newtheoremstyle{exmp_cont}
{\topsep} {\topsep}%
{\upshape}
{}
{\bfseries\scshape}
{.}
{1em}
{\thmname{#1} \thmnumber{ #2}\thmnote{#3} (continued)}
\theoremstyle{exmp_cont}

\DeclareMathOperator{\diag}{diag}
\DeclareMathOperator{\sgn}{sgn}

\DeclareMathOperator*{\argmax}{arg\,max}

\newcommand*\fall{\forall\,}

\usepackage{geometry}
\usepackage{comment}

\begin{document}

\title{A Strongly Polynomial Reduction for Linear Programs over Grids}

\author{Lorenz Klaus%
 \thanks{National Institute of Informatics (NII), Japan. {\tt lorenz@nii.ac.jp} and {\tt lklaus@iaeth.ch}}
}
\date{}

\maketitle

\begin{abstract}
We investigate the duality relation between linear programs over grids (Grid-LPs) and generalized linear complementarity problems (GLCPs) with hidden K-matrices. The two problems, moreover, share their combinatorial structure with discounted Markov decision processes (MDPs). Through proposing reduction schemes for the GLCP, we obtain a strongly polynomial reduction from Grid-LPs to linear programs over cubes (Cube-LPs). As an application, we obtain a scheme to reduce discounted MDPs to their binary counterparts. This result also suggests that Cube-LPs are the key problems with respect to solvability of linear programming in strongly polynomial time. We then consider two-player stochastic games with perfect information as a natural generalization of discounted MDPs. We identify the subclass of the GLCPs with P-matrices that corresponds to these games and also provide a characterization in terms of unique-sink orientations. A strongly polynomial reduction from the games to their binary counterparts is obtained through a generalization of our reduction for Grid-LPs.
\end{abstract}


\section{Introduction}
\emph{Linear Programming} is of particular importance in mathematical optimization. In the early days, \emph{Dantzig's simplex method} \cite{Dan:Linear} was the only practical solving method for \emph{linear programs} (LPs). The method runs fast on instances arising from applications in practice, but the algorithmic complexity had been unknown. In 1972, Klee and Minty~\cite{KleMin:How-good} constructed artificial \emph{LPs over cubes} (Cube-LPs) on which Dantzig's simplex method requires an exponential number of pivot steps.  
These days, \emph{Khachiyan's ellipsoid method} \cite{Kha} and \emph{Karmarkar's interior-point method} \cite{Kar}, both originating from nonlinear optimization, are polynomial-time solving methods. However, it is an open problem  whether there exist \emph{strongly polynomial} solving methods. These are methods where the number of arithmetic operations is polynomially bounded by the number of variables and constraints. Strong interest in the development and analysis of pivoting schemes therefore persists. Simple \emph{pivot rules} for the simplex method that terminate in a polynomial number of pivot steps would yield such a strongly polynomial solving method. For almost all ever proposed deterministic pivot rules, inefficiency has been proven through artificially constructed Cube-LPs that yield a superpolynomial number of pivot steps. Analysis of randomized and history-based pivot rules, on the contrary, is difficult. In recent years, Friedmann \emph{et al.}~\cite{friedmann2011subexponential,friedmann2011subexponentialZadeh}, nevertheless, succeeded in proving inefficiency of certain rules for \emph{LPs over grids} (Grid-LPs), such as the \textsc{Random Facet} and \textsc{Zadeh}'s rule. This raises the question whether these rules are also inefficient for ordinary Cube-LPs. In other words, we would like to identify the `key problem' of linear programming, which is the subproblem that needs to be addressed in order to clarify whether there exists a strongly polynomial solving method. In this paper, we prove that Grid-LPs admit a strongly polynomial reduction to Cube-LPs. The only requirement is that the problem instances are given in some specific representation, which is completely unrelated to the underlying combinatorics. This result gives strong evidence that linear programming over cubes is the `key problem'. Through the gained insight, the number of nontrivial constraints is a main determinant of the difficulty of linear programming with respect to pivoting algorithms. As an application, \emph{discounted Markov decision processes} (MDPs) admit formulations as Cube-LPs.

The result is obtained through a duality theory between Grid-LPs and \emph{generalized linear complementarity problems} (GLCPs) with \emph{hidden K-matrices}, which is due to Mangasarian \cite{Man:Linear}. The GLCP is a mathematical framework that unifies many optimization problems, such as linear and \emph{convex quadratic programming} \cite{Lemke:Bimatrix}, \emph{bimatrix games} \cite{Lemke:Bimatrix}, and \emph{stochastic games} \cite{gartner2005simple, svensson2007linear, jurdzinski2008simple}. Murty's book \cite{Mur:Linear} and the monograph by Cottle, Pang, and Stone \cite{CotPanSto:LCP} give a survey of applications and provide an in-depth study of the ordinary \emph{linear complementarity problem} (LCP). The generalization to GLCPs is due to Cottle and Dantzig \cite{cottle1970generalization}. The decision problem whether a general LCP has a solution is NP-complete \cite{Chung:Hardness}. LCPs with \emph{P-matrices} (P-LCPs) have a unique solution \cite{Mur:On-the-number}, and theoretical results suggest the existence of a polynomial-time solving method. The superclass of LCPs with \emph{sufficient matrices}, for instance, is in $\text{NP} \cap \text{coNP}$~\cite{FukTer:Linear}. Moreover, Megiddo \cite{Meg:A-Note-on-the-Complexity} proved that if the P-LCP is NP-hard, then $\text{NP} = \text{coNP}$. Despite these promising facts, no efficient solving method is known. The dual problem of linear programming over cubes, which is the hidden K-LCP, builds a proper subclass of the P-LCP. In order to reduce Grid-LPs to Cube-LPs, the basic idea, therefore, is to provide a reduction from GLCPs to LCPs that preserves the hidden K-property.

The exposition starts with a short preliminary section. It continues with a detailed investigation of the P, K, and hidden K-property of block matrices. The emphasis is on dual characterizations in particular. This research actually originated from the study of these matrix classes in the combinatorial setting of \emph{oriented matroids} \cite{PhD-Klaus}. In Section~\ref{sec:HidKGLCPLP}, the duality theory that relates Grid-LPs to hidden K-GLCPs is developed. We are especially interested in \emph{unique-sink orientations} (USOs), which provide a combinatorial model to illustrate the behavior of pivoting methods in terms of digraphs~\cite{StiWat:Digraph-models}. USOs arising from Grid-LPs are shown to be the same as the USOs arising from hidden K-GLCPs. Section~\ref{sec:DiscMDPs} is devoted to discounted Markov decision process (MDPs). These are the most general single-player stochastic games with perfect information. They admit formulations as Grid-LPs and, through duality, also as (hidden) K-GLCPs. Conversely, every hidden K-GLCP describes a discounted MDP. The three problems are in fact combinatorially equivalent.

The next sections are devoted to reduction schemes for several classes of GLCPs. In Section~\ref{sec:HidKKRelation}, we investigate the relation between hidden K-GLCPs and K-GLCPs in detail. The former can be converted to the latter. The dimension of the problem is preserved whereas the size of each block increases by one. As a consequence, every USO arising from a Grid-LP is fully contained in some USO arising from a K-GLCP. Through this observation, the K-GLCP cannot be considered as trivial to solve---contrary to the ordinary K-LCP, which admits strongly polynomial pivoting methods \cite{FonFukGar:Pivoting}. In Section~\ref{sec:PRed}, we propose a strongly polynomial reduction from the P-GLCP to the ordinary P-LCP. Unfortunately, the hidden K-property is not preserved. This issue can be fixed through an adapted reduction scheme. First, in Section~\ref{sec:KReduction}, we reduce arbitrary K-GLCPs to K-GLCPs with blocks of size at most two. This reduction together with the results from Section~\ref{sec:HidKKRelation} then yields a reduction from the hidden K-GLCP to the ordinary hidden K-LCP. Through duality, a reduction from Grid-LPs to Cube-LPs is obtained for free. The main result is presented in Section \ref{sec:HidKReduction}.

In Sections \ref{sec:TwoPlayerGames} and \ref{sec:Red2PG}, we provide a short discussion of two-player stochastic games with perfect information and their formulation as GLCPs. The aim is basically to generalize our results obtained for discounted MDPs. First, we identify the subclass of P-GLCPs that corresponds to these games and also provide a characterization in terms of USOs. Finally, we propose a strongly polynomial reduction from the games to their binary counterparts  through a generalization of the reduction scheme for Grid-LPs.


\section{Preliminaries} \label{sec:Prelim}

A (vertical) \emph{block matrix} is a matrix
$$M = \begin{bmatrix} M^1 \\ M^2 \\ \vdots \\ M^n \end{bmatrix} \in \mathbb{R}^{m \times n}$$
consisting of $n$ vertically aligned blocks, where each block $M^j$ has $n$ columns and an arbitrary number $b_j \in \mathbb{N}$ of rows. Block matrix $M$ is of \emph{type} $b:=(b_j)_{j \in [n]}$. 

The $i$th row in block $M^j$ is denoted by $m^j_{i \cdot}$; the $k$th element in the row by $m^j_{ik}$. Let
$$N(b) := \{ (j,i) : \, j \in [n] \text{ and } i \in [b_j] \}$$
denote the row indexes of $M$. An index set $B \subseteq N(b)$ is \emph{complementary} if for each $j \in [n]$ at most one $(j,i)$ for $i \in [b_j]$ is in $B$. If additionally $|B|=n$, then $B$ is \emph{maximal complementary}. The index set $\overline{B}:=N(b) \backslash B$ is the \emph{complement} of $B$. Sometimes, when speaking about a basis $B$, we denote the complement by $N$ instead of $\overline{B}$. Two subsets $B$ and $C$ of $N(b)$ are \emph{adjacent} if they differ in exactly one element. A \emph{representative submatrix} of $M$ is an $n \times n$ submatrix $M_{B}$ for some maximal complementary $B \subseteq N(b)$. Two representative submatrices $M_B$ and $M_{C}$ are \emph{adjacent} if $B$ and $C$ are adjacent. For a value $c \in \mathbb{R}$, let $\mathbf{c}$ denote a vector whose every entry is $c$. The dimension of $\mathbf{c}$ depends on the context.

For any $X \in \mathbb{R}^{n \times n}$, whose $j$th row is denoted by $x_j$, let $[M|X]$ denote the block matrix obtained from $M$ by extending each block $M^j$ by the row vector $x_j$, where each $x_j$ becomes the new last row of block $j$ in $[M|X]$, whose index is ($j, b_j + 1$).  Block matrix $[M|X]$ is of type $b + \mathbf{1}$.

Let $E(b)$ denote the block matrix of type $b \in \mathbb{N}^n$ whose every representative submatrix is the identity matrix.

Similar notational conventions are in use for horizontal block matrices and block vectors. Speaking of a block matrix, we usually refer to a vertical block matrix, unless stated otherwise.

\subsubsection*{Generalized linear complementarity problem}

Let $M \in \mathbb{R}^{m \times n}$ be a block matrix of type $b \in \mathbb{N}^n$ and $q \in \mathbb{R}^m$. The \emph{generalized linear complementarity problem} (GLCP) is to find a vector $z \in \mathbb{R}^n$ and a block vector $w \in \mathbb{R}^m$ of type $b$ such that
\begin{subequations} \label{eq:LCP}
 \begin{align}
  w-Mz&=q, \label{lcpvinvi} \\
  w,z &\geq \mathbf{0}, \label{lcpvposi} \\
  z^j \prod_{j=1}^{b_j} w^j_i &= 0 \quad \fall j \in [n]. \label{lcpv2compi}
 \end{align}
\end{subequations}
The GLCP($M,q$) is of \emph{type} $b$. A pair $(w,z)$ that satisfies \eqref{lcpvinvi} and the \emph{nonnegativity} condition \eqref{lcpvposi} is \emph{feasible}. If the pair additionally satisfies the \emph{complementarity} condition \eqref{lcpv2compi}, then it is a \emph{solution}. A \emph{solution basis} is a set $N = N(b + \mathbf{1}) \backslash B$ for some maximal complementary $B \subseteq N(b + \mathbf{1})$ such that the system $[I|-M]_N x = q$, $x \geq \mathbf{0}$ is feasible, where the identity matrix $I$ is supposed to be a horizontal block matrix of type $b$.

 GLCPs with special properties are in theory and practice likewise important. In this exposition, we mainly consider GLCPs with block P-matrices and interesting subclasses.

\begin{defn}
 A \emph{P-GLCP} is a GLCP($M,q$) where $M$ is a block P-matrix.
\end{defn}

A P-GLCP of type $\mathbf{1}$ is a P-LCP. Analogous notions are in use for other (block) matrix classes.

\subsubsection*{Grids and unique-sink orientations}

The grid of type $b \in \mathbb{N}^n$, denoted by G$(b)$, is the undirected Graph $(V,E)$ with
\begin{align*}
 V & := \left\{ B \subseteq N(b) : \, B \text{ is maximal complementary} \right\} \text{ and } \\
 E & := \left\{ \{B,C\} : \, B, C \in V \text{ are adjacent} \right\}.
\end{align*}
The \emph{dimension} of the grid G$(b)$ is $n$. In cases where the actual block sizes $b_j$ for $j \in [n]$ are not important, we speak of an $n$-grid. A \emph{subgrid} of the grid G$(b)$ is a subgraph induced by any subset of $N(b)$. The grid G($\mathbf{2}$) of dimension $n$ is the \emph{$n$-cube}.

A \emph{unique-sink orientation} (USO) of any grid G$(b)$ is an orientation of the edges such that every subgrid has a unique local \emph{sink}, which is a vertex with no outgoing edges. The grid G$(b)$ is considered to be a subgrid of itself. Hence, there is a unique global sink.

Stickney and Watson \cite{StiWat:Digraph-models} modeled \emph{simple principal pivoting methods} for the P-LCP as path-following algorithms on USOs of $n$-cubes. The global sink corresponds to the unique solution to the P-LCP. G\"artner, Morris, and R\"uest \cite{GarMorRus:Unique} generalized this result to P-GLCPs and USOs of grids.

Let any P-GLCP($M,q$) of type $b \in \mathbb{N}^n$ be given. Let $B$ and $C$ be any two adjacent vertices of the grid G($b + \mathbf{1}$). Suppose that $C=(B \backslash \{(j,i)\}) \cup \{(j,k)\}$, where $j \in [n]$ and $i,k \in [b_j+1]$. The orientation of the edge $\{B,C\}$ in the arising USO of the grid G($b + \mathbf{1}$) is determined by
$$B \rightarrow C :\Leftrightarrow ([ I | -M]^{-1}_N q)^j_k < 0,$$
where $N:=N(b + \mathbf{1}) \backslash B$. The block vector $[ I | -M]^{-1}_N q$ is indexed by $N$. We presuppose nondegeneracy, otherwise the orientation of some edges will be undetermined.

\begin{defn}
 A \emph{P-USO} is a USO of a grid that arises from a P-GLCP.
\end{defn}
Analogous notions are in use for subclasses of the P-GLCP. The model of USOs is generalizing, which follows from the observation that P-USOs additionally satisfy the Holt-Klee condition \cite{GarMorRus:Unique} and some counting results on USO classes \cite{FonGarKlaSpr:CountingUSOs}.

\subsubsection*{Linear programs over grids}

Consider a system
$$M^Tu \leq p \text{ and } u \geq \mathbf{0}$$
for any block matrix $M \in \mathbb{R}^{m \times n}$ of type $b \in \mathbb{N}^n$ and a vector $p \in \mathbb{R}^n$, which defines a polyhedron in $\mathbb{R}^m$ with $n + m$ facets. The feasible region of the system is \emph{combinatorially equivalent to the grid G$(b + \mathbf{1})$} if every maximal complementary $B \subseteq N(b + \mathbf{1})$ is a nondegenerate \emph{feasible basis}; i.e., every representative submatrix $[M^T|I]_B$ is nonsingular and such that $[M^T|I]_B^{-1}p > \mathbf{0}$.

A \emph{Grid-LP} is an LP
 \begin{equation*}
  \begin{aligned}
	 & \min  & q^T u & \\
	 & \text{ s.t.} & M^Tu & \leq p \\
	 &              & u & \geq \mathbf{0}
	\end{aligned}
 \end{equation*}
for $q \in \mathbb{R}^m$ whose feasible region is combinatorially equivalent to the grid G($b + \mathbf{1}$). In case of $b=\mathbf{1}$, we speak of a \emph{Cube-LP}. Every Grid-LP is obviously primal nondegenerate. If the Grid-LP is also dual nondegenerate, the objective function induces a USO of the grid G$(b + \mathbf{1})$. The orientation of the edge connecting two adjacent vertices $B$ and $C:=(B \backslash \{(j,i)\}) \cup \{(j,k)\}$, where $j \in [n]$ and $i,k \in [b_j+1]$, is determined by
$$B \rightarrow C :\Leftrightarrow (c^T_N - c^T_B A_B^{-1}A_N)^j_k < 0$$
for $A:=[M^T|I]$, $c:=[q|\mathbf{0}]$, and $N:=N(b+ \mathbf{1}) \backslash B$. The block row vector $c^T_N - c^T_B A_B^{-1}A_N$, which is indexed by $N$, is the \emph{reduced cost vector} with respect to basis $B$.

\begin{defn}
 An \emph{LP-USO} is a USO of a grid that arises from a Grid-LP.
\end{defn}

\section{Important classes of block matrices} \label{sec:BlockMatrices}

This exposition starts with a discussion of the P-property for block matrices and then proceeds to other properties, such as the Z and K-property. All definitions are given in terms of vertical block matrices.

\subsubsection*{P-property}

A \emph{square P-matrix} is an $n \times n$ matrix whose principal minors are all positive. The P-property is preserved under taking the transpose and inverse of square matrices~\cite{tsatso:generating}. The property extends straightforwardly to block matrices.

\begin{defn}
 A \emph{block P-matrix} is a block matrix whose every representative submatrix is a square P-matrix.
\end{defn}

\begin{thm} \label{thm:altBlockP}
 For a block matrix $M \in \mathbb{R}^{m \times n}$ of type $b \in \mathbb{N}^n$, the following are equivalent.
 \vspace{+2mm}
 \begin{compactenum}\itemsep2mm
  \item[\rm(a)] Block matrix $M$ is a P-matrix.
	\item[\rm(b)] Every nonzero $x \in \mathbb{R}^{n}$ satisfies $x_j(Mx)^j_i > 0$ for some $j \in [n]$ and each $i \in [b_j]$.
	\item[\rm(c)] For every nonzero $y \in \mathbb{R}^{m}$ of type $b$, there exists $j \in [n]$ such that either $y^j_i(-M^Ty)_j < 0$ for some $i \in [b_j]$ or $y^j_i y^j_k < 0$ for some $i,k \in [b_j]$.
	
	\item[\rm(d)] For every $\sigma \in \{-1 ,+1\}^n$, there exists a vector $x \in \mathbb{R}^n$ such that $\sigma_j x_j > 0$ and $\sigma_j (Mx)^j > \mathbf{0}$ for each $j \in [n]$.
	\item[\rm(e)] For every $\sigma \in \{-1 ,+1\}^{m+n}$ of type $b + \mathbf{1}$ with $\sigma^j_i \sigma^j_k < 0$ for every $j \in [n]$ with some $i,k \in [b_j +1]$, there exists a vector $y \in \mathbb{R}^{m}$ such that for every $j \in [n]$, we have $\sigma^j_i y^j_i > 0$ for all $i \in [b_j]$ and $\sigma^j_{b_j + 1} (-M^Ty)_j > 0$.
	\item[\rm(f)] Each GLCP(M,q) for $q \in \mathbb{R}^m$ has exactly one solution.
 \end{compactenum}
\end{thm}

For square matrices, equivalence of (a) and (b) is due to Fiedler and Pt\'{a}k \cite{FiePta:On-matrices-with}. The equivalence extends to block matrices. Both the implication $(b) \Longrightarrow (a)$ and the contrapositive of $(a) \Longrightarrow (b)$ directly follow from the square case. Since the simplest form of a \emph{principal pivot transform} (ppt) of $M$, in terms of characterization (b), corresponds to a single exchange $x_j \leftrightarrow (Mx)^j_i$, P-matrices are closed under ppts. Condition (c) is the dual statement of (b) and can be proven in the combinatorial setting of oriented matroids. Characterizations (d) and (e) are dual to each other and can as well be proven using oriented matroid theory. A proof for the square case is contained in \cite{FonFukKla:Combinatorial}. Characterization (f) connects the P-property to the linear complementarity theory \cite{cottle1970generalization, habetler1995existence}.

Positive row and column scaling operations preserve the P-property.

\begin{lem} \label{lem:PGLCPScaling}
 Let $M \in \mathbb{R}^{m \times n}$ be a P-matrix of type $b \in \mathbb{N}^n$. Let $L \in \mathbb{R}^{m \times m}$ and $H \in \mathbb{R}^{n \times n}$ be positive diagonal matrices. The P-GLCP$(M,q)$ and the P-GLCP$(LMH,Lq)$ for $q \in \mathbb{R}^m$ induce the same USO of the grid G$(b + \mathbf{1})$.
\end{lem}

\begin{proof}
 The feasible vectors $x \in \mathbb{R}^{m+n}$ to the system $[I|-M]x=q$, where $I$ is supposed to be a horizontal block matrix of type $b$, are the vectors $[Mz+q|z]$ for $z \in \mathbb{R}^n$. Similarly, the feasible vectors to the system $[I|-LMH]x=Lq$ are the vectors $[LMz+Lq|H^{-1}z]$ for $z \in \mathbb{R}^n$. For any maximal complementary $B \subset N(b+ \mathbf{1})$, we have $[Mz+q|z]_B=0$ for some unique $z$ if and only if $[LMz+Lq|H^{-1}z]_B=0$.
\end{proof}


For a P-matrix $M \in \mathbb{R}^{m \times n}$ of type $b \in \mathbb{N}^n$, the matrix $\mathbf{S}MS$, where $S$ is an $n \times n$ \emph{signature matrix}, is a P-matrix of the same type. The matrix $\mathbf{S}$ denotes the diagonal block matrix with $n$ blocks whose $j$th block is $s_{jj}I$ of dimension $b_j \times b_j$.

\subsubsection*{Z-property}

A \emph{square Z-matrix} is an $n \times n$ matrix whose off-diagonal elements are all nonpositive. A \emph{block Z-matrix} is a block matrix whose every representative submatrix is a square Z-matrix.

Ordinary Z-LCPs have many nice properties. For instance, a square matrix $M \in \mathbb{R}^{n \times n}$ is a Z-matrix if and only if whenever an LCP($M,q$) for $q \in \mathbb{R}^n$ is feasible, it has a solution that is the least element of the feasible region~\cite{tamir1974leastelementZ}. The least-element theory was generalized to block matrices by Ebiefung and Kostreva \cite{ebiefung1997generalized}. 

Chandrasekaran's method for Z-LCPs \cite{Cha:A-special} terminates in linear number in $n$ of iterations with either a solution or a certificate that no solution exists. The method generalizes to the Z-GLCP but is not polynomial anymore \cite{ebiefung1997algorithm}.

\subsubsection*{K-property}

A \emph{square K-matrix} is a square Z-matrix that is also a P-matrix.

\begin{defn}
 A \emph{block K-matrix} is a block matrix whose every representative submatrix is a square K-matrix.
\end{defn}


\begin{thm} \label{thm:eqBlockZ}
 Let $M \in \mathbb{R}^{m \times n}$ be a block Z-matrix of type $b \in \mathbb{N}^n$. The following statements are equivalent.
 \vspace{+2mm}
 \begin{compactenum}[\rm(a)]\itemsep2mm
  \item Block matrix $M$ is a P-matrix.
  \item There exists $x \in \mathbb{R}^n > \mathbf{0}$ with $Mx > \mathbf{0}$.
	\item Every representative matrix $C$ of $M$ is nonsingular and such that $C^{-1} \geq \mathbf{0}$.
	\item For every $p \in \mathbb{R}^n > \mathbf{0}$, the feasible region of $M^Tu \leq p$ with $u \geq \mathbf{0}$ is combinatorially equivalent to the grid G$(b + \mathbf{1})$.
	\item There exists $s \in \mathbb{R} > 0$ such that every representative matrix of $M$ can be represented as $sI-S$, where $s > \rho(S)$ and $S \geq \mathbf{0}$.
 \end{compactenum}
\end{thm}

For square matrices, equivalence of (a), (b), and (c) is due to Fiedler and Pt\'{a}k \cite{FiePta:On-matrices-with}. Equivalence of (a) and (c) for block matrices directly follows from the square case. Equivalence of (a), (b), and (d) for block matrices follows from Theorem \ref{thm:eqBlockHiddenZ} below. Condition (d) is a dual characterization. In characterization (e), which has been proposed by Ostrowski \cite{ostrowskiKmatrix} for square matrices, $\rho(\cdot)$ denotes the spectral radius. A proof for block matrices is immediately obtained.

\begin{defn} \label{lem:stochK}
 A \emph{stochastic K-matrix} is a block matrix $M$ for which there exists $\gamma \in [0,1)$ such that every representative submatrix is of the form $I - \gamma S$ for some $S \geq \mathbf{0}$ with $S \mathbf{1} = \mathbf{1}$.
\end{defn}

 Every stochastic K-matrix is a K-matrix. The Z-property is obviously satisfied. Moreover, for every representative submatrix $I-\gamma S$, we have $\rho(\gamma S) < 1$ because $S$ is a \emph{rowstochastic} matrix and thus $\rho(S)=1$. Hence $I-\gamma S$ is a square K-matrix.

Every K-matrix can be transformed into a stochastic K-matrix by appropriately scaling the rows and columns. 

\begin{lem} \label{lem:NormFormComp}
 Let $M \in \mathbb{R}^{m \times n}$ be a K-matrix of type $b \in \mathbb{N}^n$. There exists positive diagonal matrices $L \in \mathbb{R}^{m \times m}$ and $H \in \mathbb{R}^{n \times n}$ such that $LMH$ is a stochastic K-matrix of type $b$.
\end{lem}

\begin{proof}
 According to (b) in Theorem \ref{thm:eqBlockZ}, there exists an $x>\mathbf{0}$ with $Mx>\mathbf{0}$. Let $H:=\diag(x)$ and $L \in \mathbb{R}^{m \times m}$ be positive diagonal matrices such that $LMH\mathbf{1}=\mathbf{c}$ for any constant $c > 0$; here, we use that $Mx > \mathbf{0}$. Matrix $LMH$ is obviously a block K-matrix of the same type as $M$. Let $t$ be equal to the largest diagonal element that appears in any representative submatrix of $LMH$. Then, every representative submatrix can be represented as $tI - T$ for some $T \geq \mathbf{0}$. Since $(tI - T)\mathbf{1} = \mathbf{c}$, we have $t \geq c > 0$. Now, if we would have chosen $c \slash t$ instead of $c$, we would get $I - \gamma S$, where $\gamma := (t-c) \slash t$ and $S:=1 \slash (t-c) T \geq \mathbf{0}$. Note that $\gamma \in [0,1)$ and $S \mathbf{1} = \mathbf{1}$.
\end{proof}

We call $LMH$ a \emph{stochastic form} of the K-matrix $M$. Note that stochastic forms are not uniquely determined. Consider, for instance, the identity matrix $I$ of any order. Then $I=I-0I$ but also $LI=I-1/2I$ for $L:=\diag(1/2, \ldots, 1/2)$. The construction scheme outlined in the proof of Lemma~\ref{lem:NormFormComp} computes a stochastic form whose factor is minimal for fixed $x$. At the moment, it is not clear how exactly to find a stochastic form with overall minimal factor. See also the remarks accompanying Lemma~\ref{cor:StochFormHidK} below.  



\begin{defn}
 A \emph{stochastic K-GLCP} is a GLCP$(M,q)$ with a stochastic K-matrix $M$.
\end{defn}

Note that, by Lemma \ref{lem:NormFormComp} together with Lemma \ref{lem:PGLCPScaling}, every K-USO of a grid is realized by some stochastic K-GLCP.

Every principal pivoting algorithm solves the ordinary K-LCP of order $n$ in at most $2n$ pivot steps, regardless of the applied pivot rule and the initial complementary basis~ \cite{FonFukGar:Pivoting}. The K-GLCP, on the contrary, is at least as difficult as linear programming over grids. See Sections \ref{sec:HidKGLCPLP}, \ref{sec:DiscMDPs}, and \ref{sec:HidKKRelation} for details.

\subsubsection*{Hidden Z-property}

Mangasarian~\cite{Man:Linear} proposed the hidden Z-property for square matrices in order to solve certain LCPs as LPs. See also Section \ref{sec:HidKGLCPLP} for further details.
Here, we directly proceed with the generalization to block matrices, which is due to Mohan and Neogy~\cite{mohan1997vertical}.

\begin{defn}
 A \emph{block hidden Z-matrix} is a block matrix $M \in \mathbb{R}^{m \times n}$ of type $b \in \mathbb{N}^n$ for which there exist a square Z-matrix $X \in \mathbb{R}^{n \times n}$ and a block Z-matrix $Y \in \mathbb{R}^{m \times n}$ of type $b$ such that
 \begin{equation} \label{eq:blockHiddenK}
  \begin{aligned}
	   & MX=Y \text{ and } \\
     & X^T r + Y^T s > \mathbf{0} \text{ for some } r \in \mathbb{R}^n \geq \mathbf{0} \text{ and } s \in \mathbb{R}^m \geq \mathbf{0}.
    \end{aligned}
  \end{equation}
\end{defn}

The tuple $(X,Y,r,s)$ is a \emph{hidden Z-witness} of $M$. Block Z-matrix $Y$ is redundant, and thus may be omitted. Hidden Z-witnesses are not uniquely determined.

\begin{lem} \label{lem:WitnessPropScaling}
 Let $M \in \mathbb{R}^{m \times n}$ be a hidden Z-matrix with witness $(X,Y,r,s)$. A tuple $(XD,YD,r,s)$, where $D:=\diag (d)$ for any $d \in \mathbb{R}^n > \mathbf{0}$, is an alternative hidden Z-witness of $M$.
\end{lem}

\begin{lem}[\cite{mohan1997vertical}] \label{lem:WitnessProp}
 Let $(X,Y,r,s)$ be a hidden Z-witness of a block matrix. The square Z-matrix $X$ is nonsingular and some representative submatrix of $[Y|X]$ is a square K-matrix.
\end{lem}

\begin{proof}
 Consider the system $[Y^T|X^T]y= X^T r + Y^T s$, which is obviously feasible. Since the right-hand side is strictly positive and $[Y^T|X^T]$ is a Z-matrix, every feasible basis $B \subseteq N(b+\mathbf{1})$ is maximal complementary and nondegenerate. Hence, the corresponding representative submatrices $[Y^T|X^T]_B$ satisfy condition (b) in Theorem \ref{thm:eqBlockZ}, and thus are K-matrices. Now, suppose that $X$ is singular, i.e., we have $Xv=\mathbf{0}$ for some $v \neq \mathbf{0}$. Then $MXv=Yv=\mathbf{0}$, and thus $[Y|X]_Bv=\mathbf{0}$, which is a contradiction.
\end{proof}

For a block matrix $M$ of type $b$, being a hidden Z-matrix basically means that block matrix $[M^T|I]$ of type $b + \mathbf{1}$ behaves like a horizontal Z-matrix---matrices $[M^T|I]$ and $X^T[M^T|I]=[Y^T|X^T]$ have identical null and row space, where the latter matrix satisfies the Z-property.

Next, we propose a dual characterization, which is probably more illustrative than the original definition.

\begin{thm} \label{thm:hiddenZDualChar}
 A block matrix $M \in \mathbb{R}^{m \times n}$ of type $b \in \mathbb{N}^n$ is a hidden Z-matrix if and only if there exists $p \in \mathbb{R}^{n}$ such that the system $M^T u \leq p$, $u \geq \mathbf{0}$ is feasible and for every feasible block vector $u \in \mathbb{R}^{m}$ of type $b$, either $u^j \neq \mathbf{0}$ or $(M^T u)_j < p_j$ for each $j \in [n]$. Moreover, the vectors $X^{-T}v$ for $v \in \mathbb{R}^n > \mathbf{0}$, where ($X,Y,r,s$) is any hidden Z-witness of $M$, are the valid choices for such $p$.
\end{thm}

\begin{proof} $\Longrightarrow$. Let $(X,Y,r,s)$ be a hidden Z-witness of $M$, where $X=(x_j)_{j \in [n]}$ and $Y=(y_j)_{j \in [n]}$ are both to be understood columnwise. Consider the vector subspace
$$
V(M,p):=\left\lbrace x \in \mathbb{R}^{m + n + 1} : \, \begin{bmatrix} I & -M & \mathbf{0} \\ \mathbf{0}^T & p^T & 1 \end{bmatrix} x = \mathbf{0} \right\rbrace
$$
with $p:=X^{-T}v$ for any $v \in \mathbb{R}^n > \mathbf{0}$. Note that $V(M,p)$ contains for each $j \in [n]$ the vector $(-y_j^T, -x_j^T, v_j)^T$. By Motzkin's theorem of the alternative, the orthogonal vector subspace
$$
V(M,p)^\perp=\left\lbrace y \in \mathbb{R}^{m + n + 1} : \, \begin{bmatrix} M^T & I & -p \end{bmatrix} y = \mathbf{0} \right\rbrace
$$
contains no vector $(w^T, z^T, 1)^T$, where $w \in \mathbb{R}^m$ is of type $b$ and $z$ is in $\mathbb{R}^n$, with $w,z \geq \mathbf{0}$ while for some $j \in [n]$ both $w^j=\mathbf{0}$ and $z_j=0$. Since $X$ is nonsingular, a vector $y \in \mathbb{R}^{m + n +1}$ is in $V(M,p)^\perp$ if and only if
$\begin{bmatrix} Y^T & X^T & -v \end{bmatrix} y = \mathbf{0}.$
Since $[Y^T|X^T]$ is a K-matrix and $v > \mathbf{0}$, it follows that $M^T u \leq p$ with $u \geq \mathbf{0}$ is feasible.

$\Longleftarrow$. Let $M$ and $p$ be such that $V(M,p)^\perp$ is as supposed to be. By Motzkin's theorem of the alternative, the orthogonal vector subspace $V(M,p)$ contains for each $j \in [n]$ a vector $(y_j^T, x_j^T, 1)^T$, where $y_j \in \mathbb{R}^m$ is of type $b$ and $x_j$ is in $\mathbb{R}^n$, with $(y_j)^k \geq \mathbf{0}$ and $(x_j)_k \geq 0$ for every $k \neq j$. Let $X:=(-x_j)_{j \in [n]}$ and $Y:=(-y_j)_{j \in [n]}$, both defined columnwise. Matrices $X$ and $Y$ are Z-matrices and, by the structure of $V(M,p)$, we have $MX=Y$ and $p^TX = \mathbf{1}^T$. By assumption, vector subspace $V(M,p)^\perp$ contains a nonnegative vector $(s^T,r^T,1)^T$. Hence $p=r+M^Ts$, and thus $p^TX=r^TX+s^TMX=r^TX+s^TY$. Thus $r^TX+s^TY = \mathbf{1}^T> \mathbf{0}^T$. The tuple ($X,Y,r,s$) is a hidden Z-witness of $M$.
\end{proof}

Since the simplest form of a ppt, in terms of Theorem \ref{thm:hiddenZDualChar}, corresponds to an exchange $u^j_i \leftrightarrow p_j - (M^Tu)_j$, the hidden Z-matrices are closed under ppts. Every Z-matrix $M$ is a hidden Z-matrix---the tuple $(I, M, \mathbf{1}, \mathbf{0})$ is a witness. The Z-property is not preserved under taking ppts.

\subsubsection*{Hidden K-property}

The definition of the hidden K-property is now straightforward.

\begin{defn}
 A \emph{block hidden K-matrix} is a block hidden Z-matrix that is also a block P-matrix.
\end{defn}

\begin{thm} Let $M \in \mathbb{R}^{m \times n}$ be a block hidden Z-matrix with witness $(X,Y,r,s)$ of type $b \in \mathbb{N}^n$. The following statements are equivalent. \label{thm:eqBlockHiddenZ}
 \vspace{+2mm}
 \begin{compactenum}[\rm(a)]\itemsep2mm
  \item Block matrix $M$ is a P-matrix.
  \item There exists $x \in \mathbb{R}^n > \mathbf{0}$ with $Mx > \mathbf{0}$.
	\item Block matrix $[Y | X]$ is a K-matrix.
	\item There exists $p \in \mathbb{R}^n > \mathbf{0}$ such that the feasible region of $M^Tu \leq p$ with $u \geq \mathbf{0}$ is combinatorially equivalent to the grid G$(b + \mathbf{1})$.
 \end{compactenum}
\end{thm}

Equivalence of (a), (b), and (c) is proven in \cite{mohan1997vertical}.

\begin{proof}[Proof of Theorem \ref{thm:eqBlockHiddenZ}]



 $(c) \Longrightarrow (d).$ Let $p:=X^{-T} v > \mathbf{0}$ for any $v \in \mathbb{R}^n > \mathbf{0}$. Note that $M^T=X^{-T}Y^T$. Since $X$ is nonsingular, the feasible regions of $[M^T|I]y=p$ and $[Y^T | X^T] y = X^Tp$ are equivalent. Every representative submatrix $C$ of $[Y^T | X^T]$ is a K-matrix. Since $C^{-1} \geq \mathbf{0}$ is also a P-matrix, it follows that $C^{-1}X^Tp = C^{-1} v > \mathbf{0}$.

 $(d) \Longrightarrow (a).$ Pick $p \in \mathbb{R}^n$ such that the feasible region of $M^T u \leq p$ with $u \geq \mathbf{0}$ is combinatorially equivalent to the grid G$(b + \mathbf{1})$. Let $C=(c_j)_{j \in [n]}$ and $D=(d_j)_{j \in [n]}$, both to be understood columnwise, be two adjacent representative submatrices of $[M^T | I]$. Let $j \in [n]$ be the unique index with $c_j \neq d_j$. According to Cramer's rule, we have $(C^{-1}p)_j = \det C[c_j \rightarrow p] \slash \det C$ and $(D^{-1}p)_j = \det D[d_j \rightarrow p] \slash \det D$. Since $C[c_j \rightarrow p]=D[d_j \rightarrow p]$ and both $(C^{-1}p)_j > 0$ and $(D^{-1}p)_j > 0$, we have $ \sgn \det C= \sgn \det D$. Hence, the determinants of any two adjacent representative submatrices have the same nonzero sign, which must be positive because $I$ is a representative submatrix. The matrix $[M^T | I]$ is a P-matrix, and thus $M^T$ is a P-matrix.
\end{proof}

Since every Z-matrix is a hidden Z-matrix, Theorem \ref{thm:eqBlockHiddenZ} proves equivalence of (a), (b), and (d) in Theorem \ref{thm:eqBlockZ}. Every K-matrix is a hidden K-matrix.
%




Condition (d) is actually a dual characterization of the hidden K-property, which has originally been proposed for square matrices in a different context \cite{PanCha:Linear, MorLaw:Geometric}. We used the characterization in order to generalize the hidden K-property in the setting of oriented matroids \cite{PhD-Klaus}.

\begin{thm} \label{thm:HiddKDual}
 A matrix $M \in \mathbb{R}^{m \times n}$ of type $b \in \mathbb{N}^n$ is a hidden K-matrix if and only if there exists $p \in \mathbb{R}^{n}$ such that the feasible region of $M^T u \leq p$ with $u \geq \mathbf{0}$ is combinatorially equivalent to the grid G$(b + \mathbf{1})$. The valid choices for such $p$ coincide with the choices given in Theorem \ref{thm:hiddenZDualChar}.
\end{thm}

\begin{proof}
 Necessity is given by Theorem \ref{thm:eqBlockHiddenZ}. For sufficiency, we remark that if the feasible region of the system $M^Tu \leq p$, $u \geq \mathbf{0}$ is combinatorially equivalent to the grid G($b + \mathbf{1}$), then the dual characterization of the hidden Z-property is satisfied (cf.~Theorem \ref{thm:hiddenZDualChar}). Hence $M$ is a hidden Z-matrix, and Theorem \ref{thm:eqBlockHiddenZ} applies once more.
\end{proof}

We would like to find a witness that verifies the hidden K-property of a matrix $M$ with ease. In principle, any hidden Z-witness $(X,Y,r,s)$ of $M$ would be valid choice. One just verifies the hidden Z-property of $M$ and finds a vector $x > \mathbf{0}$ such that $[Y|X]x > \mathbf{0}$, which involves a linear feasibility problem. Then $[Y|X]$ is a K-matrix and Theorem \ref{thm:eqBlockHiddenZ} proves $M$ to be a hidden K-matrix. Next, we propose a hidden K-witness that avoids solving an LP for verification tasks.

\begin{prop} [\cite{Rus:The-P-Matrix-Linear}] \label{prop:HidKWitChar}
 A block matrix $M \in \mathbb{R}^{m \times n}$ is a hidden K-matrix if and only if there exists a square Z-matrix $X \in \mathbb{R}^{n \times n}$ and a Z-matrix $Y \in \mathbb{R}^{m \times n}$ of the same type as $M$ such that $MX=Y$ and $[Y|X]\mathbf{1} > \mathbf{0}$.
\end{prop}

\begin{proof}
$\Longrightarrow$. Let ($X,Y,r,s$) be any hidden Z-witness of $M$. Since $[Y|X]$ is a K-matrix, there exists $x \in \mathbb{R}^n > \mathbf{0}$ with $[Y|X]x > \mathbf{0}$. Then $(XD,YD,r,s)$ with $D:= \diag(x)$ is an alternative hidden Z-witness of $M$ that satisfies $[YD|XD]\mathbf{1} > \mathbf{0}$.

$\Longleftarrow$. Let $X$ and $Y$ be Z-matrices such that $MX=Y$ and $[Y|X]\mathbf{1} > \mathbf{0}$. Because $X^T$ is a K-matrix, there exists $r \in \mathbb{R}^n > \mathbf{0}$ with $X^Tr > \mathbf{0}$. The tuple ($X,Y,r,\mathbf{0}$) is a hidden Z-witness of $M$. Since $[Y|X]$ is K-matrix, the matrix $M$ is also a P-matrix.
\end{proof}

A \emph{proper hidden K-witness} of a block matrix $M$ is any pair $(X, Y)$ that satisfies the conditions in Proposition \ref{prop:HidKWitChar}. Block K-matrix $Y$ may be omitted.

The following result is a generalization of Lemma~\ref{lem:NormFormComp}.

\begin{lem} \label{cor:StochFormHidK}
 Let $M \in \mathbb{R}^{m \times n}$ be a hidden K-matrix with proper witness $(X,Y)$. There exist positive diagonal matrices $L \in \mathbb{R}^{m \times m}$ and $H \in \mathbb{R}^{n \times n}$ such that $LMH$ is a hidden K-matrix with proper witness $(H^{-1}X,LY)$ and, moreover, the matrix $[LY|H^{-1}X]$ is a stochastic K-matrix.
\end{lem}

\begin{proof}
Let $L \in \mathbb{R}^{m \times m}$ and $H \in \mathbb{R}^{n \times n}$ be positive diagonal matrices such that $[LY |H^{-1}X] \mathbf{1} =\mathbf{c}$ for any positive constant $c$. In other words, we scale the rows such that the entries in each row of $[LY |H^{-1}X]$ sum up to the same positive constant $c$. Here, we require that $X \mathbf{1} > \mathbf{0}$ and $Y \mathbf{1} > \mathbf{0}$. Let $t > 0$ be equal to the largest diagonal element in any representative submatrix of $[LY |H^{-1}X]$. Since $X$ and $Y$ are K-matrices, every representative submatrix can be represented as $tI - T$ for some $T \geq \mathbf{0}$. Recall that $(tI - T) \mathbf{1}= \mathbf{c}$, and thus $t \geq c > 0$. If we would have chosen $c \slash t$ instead of $c$, then we would get $I - \gamma S$, where $\gamma := (t-c) \slash t$ and $S:=(1 \slash (t-c)) T \geq \mathbf{0}$. Note that $\gamma \in [0,1)$ and $S \mathbf{1} = \mathbf{1}$. Consequently, the matrix $[LY|H^{-1}X]$ is a stochastic K-matrix; and since $LMH(H^{-1}X)=LMX=LY$, the matrix $LMH$ is a hidden K-matrix with proper witness $(H^{-1}X,LY)$.
\end{proof}

Deciding for a given matrix with known block type whether it has the hidden K-property can be done in polynomial time.\label{page:CompHidKWit} The conditions in Proposition \ref{prop:HidKWitChar} can be formulated as a linear feasibility problem. Hence, we can assume that there always exists a proper hidden K-witness with polynomial binary encoding length. In practice, we often like to compute a proper witness $(X,Y)$ of a given hidden K-matrix $M$ of type $b$ such that the factor of the stochastic K-matrix $[LY|H^{-1}X]$ in Lemma~\ref{cor:StochFormHidK} is as small as possible. Such a problem may be formulated as the LP~\cite{morris2012efficient}
\begin{equation} \label{eqn:minFacWit}
 \begin{aligned}
	 & \min \gamma  \\
	 & \text{ s.t. } [MX|X] \leq E(b + \mathbf{1}), \\
	 & \phantom{\text{ s.t. }} [MX|X] \mathbf{1} \geq (1 - \gamma) \mathbf{1},
	\end{aligned} 
\end{equation}
where matrix $X$ and factor $\gamma$ are the variables. The system is obviously feasible with $\gamma > 0$ as any proper witness, where at least one is supposed to exist, can be scaled by Lemma~\ref{lem:WitnessPropScaling}. Let $(\gamma^*,X^*)$ be an optimal solution to the LP~\eqref{eqn:minFacWit}. The pair $(X^*,MX^*)$ is certainly a proper hidden K-witness of $M$. The witness may then be used for Lemma~\ref{cor:StochFormHidK}. From $[MX^*|X^*] \mathbf{1} \geq (1 - \gamma^*) \mathbf{1}$ it follows that the components of the positive diagonal matrices $L$ and $H$ in the equation $[LMX^*|HX^*] \mathbf{1} = (1 - \gamma^*) \mathbf{1}$ are at most $1$. Then $t$ will be at most $1$, and thus, the factor of the final stochastic K-matrix will be at most $(1-(1-\gamma^*))\slash 1=\gamma^*$. Morris~\cite{morris2012efficient} proposed a solving scheme for the LP~\eqref{eqn:minFacWit} that exploits the close connection to discounted MDPs. The scheme is strongly polynomial if the factor $\gamma^*$ is considered to be a constant. No other, especially no strongly polynomial algorithm in the technical sense is known for the computation of hidden K-witnesses.

\section{A duality theory for Grid-LPs and hidden K-GLCPs} \label{sec:HidKGLCPLP}

We establish a duality theory for linear programming and the GLCP, which originates from Mangasarian's seminal paper \cite{Man:Linear}. Here, we are particularly interested in the relation regarding the combinatorial abstraction of USOs.

Let a Z-matrix $M \in \mathbb{R}^{m \times n}$ with witness $(X,Y,r,s)$ and a vector $q \in \mathbb{R}^{m}$ be given, both of the same type $b \in \mathbb{N}^n$.

The hidden Z-GLCP($M,q$), which is to find a vector $z \in \mathbb{R}^n$ and a vector $w \in \mathbb{R}^m$ of type $b$ such that
\begin{equation} \label{eqn:PrimalGLCP}
	w - M z = q, \qquad w, z \geq \mathbf{0}, \qquad z^j \prod_{j=1}^{b_j} w^j_i = 0 \quad \fall j \in [n], 
\end{equation}
is considered to be the primal problem. Let $p:=X^{-T}v$ for any $v \in \mathbb{R}^n > \mathbf{0}$. The LP
\begin{equation} \label{eqn:PrimalLP}
  \begin{aligned}
	 & \max  & -p^T z & \\
	 & \text{ s.t.} & -M z  & \leq q, \\
	 &              & z & \geq \mathbf{0},
	\end{aligned} 
\end{equation}
has the same feasible region as the hidden Z-GLCP \eqref{eqn:PrimalGLCP}. The dual of LP \eqref{eqn:PrimalLP}, which is
\begin{equation} \label{eqn:DualLP}
  \begin{aligned}
	 & \min  & q^T u & \\
	 & \text{ s.t.} & M^T u  & \leq p, \\
	 &              & u & \geq \mathbf{0},
	\end{aligned} 
 \end{equation}
is regarded as the dual problem. The dual problem is not unique---it is actually a bunch of LPs \eqref{eqn:DualLP}, each identified by right-hand side $p$.

\begin{prop} [\cite{mohan1997vertical}] \label{prop:HidZZGLCPSol}
If $z^*$ is a solution to any LP \eqref{eqn:PrimalLP}, then $(q+Mz^*,z^*)$ is a solution to the hidden Z-GLCP \eqref{eqn:PrimalGLCP} .
\end{prop}

\begin{proof}
The feasible region of the corresponding dual LP \eqref{eqn:DualLP} is also determined by the equation system $[Y^T|X^T]y=X^{T}p$ and $y \geq \mathbf{0}$. Since $X^{T}p > \mathbf{0}$ and $[Y^T|X^T]$ is a block Z-matrix, every feasible basis must be maximal complementary, and so is every optimal basis. Through the complementary slackness conditions for linear programming, any optimal basis to the LP \eqref{eqn:PrimalLP} must then satisfy the complementarity constraints of a GLCP
\end{proof}

It may happen that some LP \eqref{eqn:DualLP} is unbounded, then the corresponding LP \eqref{eqn:PrimalLP} is infeasible, and thus there is no solution to the hidden Z-GLCP \eqref{eqn:PrimalGLCP}.


In the remainder, we shall restrict ourselves to the case where $M$ is a hidden K-matrix.

\begin{prop} \label{prop:LPHiddKSameSol}
 Let $M \in \mathbb{R}^{m \times n}$ be a hidden K-matrix. A set $B \subseteq N(b + \mathbf{1})$ is a solution basis to any LP \eqref{eqn:DualLP} if and only if $N:=N(b + \mathbf{1}) \backslash B$ is a solution basis to the hidden K-GLCP \eqref{eqn:PrimalGLCP}.
\end{prop}

\begin{proof}
This result follows from Proposition \ref{prop:HidZZGLCPSol} and the fact that every hidden K-GLCP has a unique solution.
\end{proof}

Next, we propose alternative expressions for the reduced cost vectors of an LP \eqref{eqn:DualLP}.

\begin{lem} \label{lem:RedCostVec}
 Consider an LP \eqref{eqn:DualLP} in normal form. For every basis $B \subseteq N(b + \mathbf{1})$, we have
$$c^T_N - c^T_B A_B^{-1} A_N = [I|-M]_N^{-1}q,$$
where $A:=[M^T|I]$, $c:=[q|\mathbf{0}]$, and $N:=N(b + \mathbf{1}) \backslash B$.
\end{lem}

\begin{proof}
First, consider the basis $\left\{(j,b_j+1) : \, j \in [n] \right\}$ to the LP \eqref{eqn:DualLP}. The reduced cost vector is $q$, which at the same time equals the basic solution to the LP \eqref{eqn:PrimalLP} with respect to the complement. More generally, the reduced cost vector with respect to any basis $B \subseteq N(b + \mathbf{1})$ is $c^T_N - c^T_B A_B^{-1} A_N$, which equals the basic solution to the LP \eqref{eqn:PrimalLP} with respect to $N:=N(b + \mathbf{1}) \backslash B$.
\end{proof}

If $M$ is a hidden K-matrix, then every LP \eqref{eqn:DualLP} is a Grid-LP by Theorem \ref{thm:HiddKDual}. Recall that a hidden K-GLCP has a unique solution, but not necessarily a unique solution basis. For the solution basis to be unique, we ask for nondegeneracy. Any hidden K-GLCP \eqref{eqn:PrimalGLCP} and the corresponding LPs \eqref{eqn:DualLP} yield the same USO, assuming nondegeneracy. The following theorem also follows from R\"ust's PhD thesis \cite{Rus:The-P-Matrix-Linear}. For the square case, we made use of it in \cite{FonGarKlaSpr:CountingUSOs} without further explanations.

\begin{thm} \label{thm:LPUSOeqGLCPUSO}
 A USO of a grid is a hidden K-USO if and only if it is an LP-USO. Assuming nondegeneracy, the following holds.
 \vspace{+2mm}
 \begin{compactenum}[\rm(i)]\itemsep2mm
  \item For a hidden K-GLCP \eqref{eqn:PrimalGLCP}, every LP \eqref{eqn:DualLP} is a Grid-LP and induces the same USO.
  \item If any LP \eqref{eqn:DualLP} is a Grid-LP, then the matrix $M$ is a hidden K-matrix and the hidden K-GLCP~\eqref{eqn:PrimalGLCP} induces the same USO.
 \end{compactenum}
\end{thm}

\begin{proof}
(i). By Theorem \ref{thm:HiddKDual}, every LP \eqref{eqn:DualLP} must be a Grid-LP. The orientations of the edges incident to a vertex $B \subseteq N(b + \mathbf{1})$ in the arising LP-USOs are determined by the corresponding reduced cost vector, which, by Lemma \ref{lem:RedCostVec}, equals $[I|-M]_N^{-1}q$, where $N:=N(b + \mathbf{1}) \backslash B$. These expressions also determine the USO induced by the hidden K-GLCP \eqref{eqn:PrimalGLCP}.

(ii). Since an LP \eqref{eqn:DualLP} is a Grid-LP, block matrix $M$ is a hidden K-matrix by Theorem \ref{thm:HiddKDual}. Lemma \ref{lem:RedCostVec} applies once more.
\end{proof}

We conclude that the hidden K-GLCP and linear programming over grids are equally difficult with respect to simplex-type methods. With respect to an arbitrary solving method, the hidden K-GLCP is at least as difficult because the right-hand side $p$ is not part of the input in the complementarity setting. A `smart' algorithm for linear programming may make use of~$p$, unlike simplex-type methods. Karmarkar's interior-point method is such a method. This leads to the question as to whether any algorithm for linear programming (over grids) that does not really make use of the right-hand side of the constraints can be a polynomial-time solving method (in any sense).

\section{On the exact relation between discounted MDPs, Grid-LPs, and (hidden) K-GLCPs} \label{sec:DiscMDPs}

 We study \emph{single-player stochastic games} with perfect information, which are well-known to admit formulations as Grid-LPs~\cite{d1963probabilistic}. We are interested in the exact relation to the GLCP. The games are known to admit formulations as (hidden) K-GLCPs~\cite{gartner2005simple, svensson2007linear, jurdzinski2008simple}, which also follows from the duality theory discussed in the previous section.  In this section, all these relations are explained in detail through investigation of discounted Markov decision processes (MDPs) as the most general variant\footnote{K-GLCP formulations for discounted MDPs have also been provided by Sumita and Kakimura (pers.~comm.).}. We will prove below that discounted MDPs are combinatorially equivalent to Grid-LPs and hidden K-GLCPs.

A \emph{Markov decision process} (MDP) is a stochastic process with discrete steps. At each step, the process is in some state $j$ and chooses from an available action $i$. The process randomly moves into some other state. The probability $p^j_{ik} \in [0,1]$ that the process moves into state $k$ is determined by the current state $j$ and chosen action $i$. A moving step has reward $r^j_i \in \mathbb{R}$ assigned, which likewise depends on the state $j$ and action $i$. For discounted MDPs, the rewards are discounted by some factor $\gamma \in [0, 1)$. To summarize, the following notations are used:
		
		\vspace{+0.3cm}
		\begin{tabular}{p{0.0cm}p{0.5cm}p{25cm}}
	   & $n$ & number of states, \\
		 & $a_j$ & number of actions available in state $j \in [n]$, \\
		 & $r^j_i$ & reward for choosing action $i \in [a_j]$ in state $j$, \\
		 & $p^j_{ik}$ & conditional probability to arrive in state $k$ for state $j$ and action $i \in [a_j]$, \\
		 & $\gamma$ & discount factor ($\gamma \in [0,1)$).\\
	  \end{tabular}
		\vspace{+0.3cm}
		
\noindent The problem is to find an optimal policy. A \emph{policy} is a function $\pi$ that specifies for each state the action to take. For a fixed policy, the MDP becomes an ordinary Markov chain. An \emph{optimal policy} is any policy that maximizes the total reward in expectation over an infinite-time horizon, where the rewards are discounted by factor $\gamma$.

Let $v^{\pi}_j$ denote the total discounted reward in expectation for initial state $j$ when applying policy $\pi$. The \emph{Bellman equations}
\begin{equation} \label{eq:Bellmann}
 v^{\pi}_j = \max_{i \in [a_j]} \left\{ r_i^j + \gamma \sum_{k=1}^n p_{ik}^j v^{\pi}_k \right\} \qquad \fall j \in [n]
\end{equation}
describe an optimality criterion for a policy $\pi$. Shapley~\cite{shapley1953stochastic} proved that the equations have a unique solution $v^*$. The optimal policies are the policies that for each state $j \in [n]$ select an action $i \in [a_j]$ in
$$\argmax_{i \in [a_j]} \left\{ r_i^j + \gamma \sum_{k=1}^n p_{ik}^j v^*_k \right\}.$$

The Bellman equations can be written as a GLCP using the approach proposed by Svensson and Vorobyov \cite{svensson2007linear}. First, we introduce the slack block variable $w \in \mathbb{R}^m$ of type $a:=(a_j)_{j \in [n]}$. The problem is then to find $v$ and $w$ such that
\begin{equation*} 
 \begin{aligned}
  v_j & = w_i^j + r_i^j + \gamma \sum_{k=1}^n p_{ik}^{j} v_k & \fall j \in [n], \fall i \in [a_j], \\
	w & \geq \mathbf{0}, & \\
	\prod_{i = 1}^{a_j} w_i^j & = 0 & \fall j \in [n].
 \end{aligned}
\end{equation*}
The complementarity problem is not yet in proper form. Note that every entry $v_j$ is lower bounded by
$$\sum_{i=0}^{\infty}\gamma^id=\frac{d}{1-\gamma} \quad \mbox{ for } d < \min \{ r_i^j : j \in [n] \text{ and } i \in [a_j]\}.$$
The GLCP can be written in proper form by substituting $z_j+d\slash(1-\gamma) $ for $v_j$ and adding the constraint $z_j \geq 0$ for every $j \in [n]$. We obtain the problem
\begin{equation*} 
 \begin{aligned}
  z_j & = w_i^j + r_i^j - d + \gamma \sum_{k=1}^n p_{ik}^{j} z_k & \fall j \in [n], \fall i \in [a_j], \\
	z, w & \geq \mathbf{0}, & \\
	z_j \prod_{i = 1}^{a_j} w_i^j & = 0 & \fall j \in [n].
 \end{aligned}
\end{equation*}

Let $P:=(p^j_{ik})_{j,k \in [n], i \in [a_j]}$ and $r:=(r^j_i)_{j \in [n], i \in [a_j]}$, which both are of type $a$. The reduction results in a GLCP($M,q$) with $M:=E(a)-\gamma P$ and $q:=-r + \mathbf{d}$. The matrix $M$ is a stochastic K-matrix of type $a$ and $q$ is strictly negative.

\begin{prop} \label{prop:MDPsAsKGLCPs}
 Every discounted MDP admits a formulation as a stochastic K-GLCP.
\end{prop}

Our reduction provides an alternative proof for the fact that the Bellman equations \eqref{eq:Bellmann} have a unique solution~\cite{shapley1953stochastic}. Moreover, if the arising K-GLCP is nondegenerate, then the MDP has a unique optimal policy.

By duality, the Grid-LP
$$\min \, (-r^T + \mathbf{d}^T) u \text{ subject to } M^Tu \leq p \text{ and } u \geq \mathbf{0}$$
for any $p > \mathbf{0}$ solves the same problem (cf. Section \ref{sec:HidKGLCPLP}). Note that the unique solution ($w,z$) to the K-GLCP is such that $z > \mathbf{0}$. Therefore, through complementary slackness, any optimal solution $u$ to the Grid-LP satisfies $M^Tu=p$. The LP
\begin{equation} \label{eq:LPFormulationDMDPMy}
 \min \, (-r^T + \mathbf{d}^T) y \text{ subject to } M^Ty = p \text{ and } y \geq \mathbf{0},
\end{equation}
which stays a Grid-LP, solves the MDP as well. We may further simplify the problem.

\begin{lem}
 A maximal complementary $B \subseteq N(a)$ is a solution basis to the Grid-LP \eqref{eq:LPFormulationDMDPMy} if and only if $B$ is a solution basis to the Grid-LP
 \begin{equation} \label{eq:LPFormulationDMDP}
  \max r^T y \text{ subject to } M^Ty = p \text{ and } y \geq \mathbf{0}.
 \end{equation}
 Moreover, the two Grid-LPs induce the same USO of the grid $G(a)$.
\end{lem}

\begin{proof}
 We conduct a sensitivity analysis. Let $B$ be any maximal complementary subset of $N(a)$. Set $B$ is a feasible basis to both Grid-LPs. Moreover, set $B$ is a solution basis to the Grid-LP \eqref{eq:LPFormulationDMDP} if and only if the corresponding reduced cost vector $-r^T_{N} + r_B^T M_B^{-T} M^T_{N}$, where $N:=N(a) \backslash B$, is nonnegative. Since $M$ is a stochastic K-matrix, we have $\mathbf{d}^TM_B^T=(1-\gamma)\mathbf{d}^T$; and thus $\mathbf{d}^T M_B^{-T} = (1 \slash (1 - \gamma)) \mathbf{d}^T$. In addition $\mathbf{d^T} M^T_N=(1-\gamma) \mathbf{d}^T$. Then
$$-r^T_{N} + r_B^T M_B^{-T} M^T_{N}=(-r^T + \mathbf{d}^T)_{N}- (- r^T + \mathbf{d}^T)_B M_B^{-T} M^T_{N}.$$
The right-hand side of the equation represents the reduced cost vector with respect to $B$ of the Grid-LP~\eqref{eq:LPFormulationDMDPMy}. Hence, the result follows.
\end{proof}

Note that the Grid-LP \eqref{eq:LPFormulationDMDP} corresponds to the Grid-LP formulation for discounted MDPs proposed by d'Epenoux \cite{d1963probabilistic}.

\begin{defn}
 A Grid-LP in \emph{stochastic form} is a Grid-LP of the form \eqref{eq:LPFormulationDMDP}.
\end{defn}

We may think of discounted MDPs and Grid-LPs in stochastic form as the same problems. Morris~\cite{morris2012efficient} observed that, in terms of simplex-type methods, solving a Grid-LP corresponds to solving some discounted MDP.

\begin{thm} \label{thm:LPUSOisMDPUSO}
Every LP-USO of a grid arises from some discounted MDP. More precisely, the USO arising from a Grid-LP
\begin{equation} \label{eqn:StartUSOGridLP}
 \min q^T u \text{ subject to }M^T u \leq p \text{ and } u \geq \mathbf{0}
\end{equation}
also arises from the discounted MDP
$$\max -[q^T L | \mathbf{0}] y \text{ subject to } [Y^TL | X^TH^{-1}] y = X^Tp \text{ and } y \geq \mathbf{0},$$
where $(X,Y)$ is a proper hidden K-witness of $M$ and $L$ and $H$ are selected as in Lemma~\ref{cor:StochFormHidK}.
\end{thm}

\begin{proof}
 By Theorem \ref{thm:LPUSOeqGLCPUSO}, the matrix $M$ has the hidden K-property and, moreover, the hidden K-GLCP$(M,q)$ induces the same USO as the Grid-LP \eqref{eqn:StartUSOGridLP}. Recall that $LMH$ is a hidden K-matrix with witness $(H^{-1}X,LY)$. A dual problem of the hidden K-GLCP$(LMH,Lq)$, which by Lemma~\ref{lem:PGLCPScaling} likewise induces the same USO, is, for instance, the Grid-LP $\min [q^T L | \mathbf{0}] y$ subject to $[HM^TL | I] y = Hp$ and $y \geq \mathbf{0}$, as $p:=X^{-T}v$ for some $v > \mathbf{0}$ and therefore $Hp=HX^{-T}v=(H^{-1}X)^{-T}v$ is a valid right-hand side. Finally, we multiply $X^TH^{-1}$ from the left to the system $[HM^TL | I] y = Hp$. Such an operation does not change the LP, as both $X$ and $H$ are nonsingular. (We also have $X^Tp \geq \mathbf{0}$ as $p=X^{-T}v$.)
\end{proof}

This reduction is slightly different from Morris' reduction~\cite{morris2012efficient}. We do not introduce an artificial state when formulating a Grid-LP as a discounted MDP, but lose in return the property that Dantzig's simplex method, which is not purely combinatorial, behaves the same for the two problems.

To summarize, discounted MDPs, Grid-LPs, and hidden K-GLCPs provide different formulations of the very same underlying problem. The collection of LP-USOs of grids characterizes the combinatorial structure of these problems. \emph{Single-switch policy iteration methods} for discounted MDPs, simplex-type methods for Grid-LPs, and simple principal pivoting methods for the hidden K-GLCP have identical algorithmic complexity (for corresponding combinatorial pivot rules).


Ye \textit{et al.~}\cite{ye2011simplex, hansen2013strategy} determined upper bounds on the number of pivot steps of Dantzig's simplex method for discounted MDPs. The bounds are in terms of the number of states $n$, the total number of actions $m$, and the discount factor $\gamma$. They likewise hold for Grid-LPs and hidden K-GLCPs. For a hidden K-GLCP$(M,q)$, where the $m \times n$ matrix $M$ has a proper hidden K-witness $(X,Y)$ such that $[Y|X]$ is a stochastic K-matrix with factor $\gamma$, Theorem~\ref{thm:LPUSOisMDPUSO} immediately gives the bound $\smash{O(\frac{(m+n)n}{1-\gamma} \log \frac{n}{1-\gamma})}$ on the number pivot steps of principal pivoting with Dantzig's pivot rule\footnote{The initial bound on the number of pivot steps of Dantzig's pivot rule for discounted MDPs is taken from~\cite{hansen2013strategy}.}. A valid bound for arbitrary hidden K-GLCPs is obtained through Morris' reduction~\cite{morris2012efficient}.

The \emph{value iteration method} for discounted MDPs transforms into a solving method for (hidden) K-GLCPs. The method even generalizes into a solving scheme for hidden Z-GLCPs~\cite{klaus2014submodularity}. Conversely, \emph{Lemke's method} provides a new solving method for discounted MDPs~\cite{fearnley2010linear}.

\medskip

The following illustrates a first application of the many reductions we presented in this section. For a given discounted MDP, it is desired to find an MDP with the same optimal policies whose discount factor is as small as possible.

\begin{thm}
 Consider a discounted MDP$(P,r,\gamma)$ of type $b \in \mathbb{N}^n$. Let $f:= \min_{j \in [n], i \in [b_j]} \{p^j_{ij}\}$. The
 \begin{equation} \label{eqn:DisFacRed}
  \text{discounted MDP}((\gamma \slash  \kappa) (P - f E(b)), r,\kappa \slash \lambda)
 \end{equation}
 for $\kappa:=\gamma (1 - f)$ and $\lambda:=1-\gamma f$ has the same optimal policies.
\end{thm}

\begin{proof}
 Let $E:=E(b)$. The initial MDP is solved by the K-GLCP$(E-\gamma P, -r + \mathbf{d})$ for sufficiently small $d < \mathbf{0}$. By Lemma~\ref{lem:PGLCPScaling}, the K-GLCP$((1\slash \lambda) (E-\gamma P), -r + \mathbf{d})$ induces the same USO. In order to verify that the latter K-GLCP represents the discounted MDP~\eqref{eqn:DisFacRed}, we observe that
 \begin{align*}
  (1 \slash \lambda) (E - \gamma P) & = (1 \slash \lambda)  (E - \gamma (f E + (P - f E))) \\
	                        & = (1 \slash \lambda)  (\lambda E + \gamma (P - f E)) \\
													& = E - (1 \slash \lambda)  \gamma (P - f E) \\
													& = E - (\kappa \slash \lambda) ( \gamma \slash \kappa) (P-f E)
 \end{align*}
 The matrix $(\gamma \slash \kappa) (P-f E)$ is rowstochastic as $P-f E \geq \mathbf{0}$ and $(\gamma \slash \kappa) (P-f E) \mathbf{1} = \mathbf{1}$. Hence, the matrix $(1 \slash \lambda) (E - \gamma P)$ is a stochastic K-matrix with factor $\kappa \slash \lambda$, which is in $[0,\gamma)$. The discount factor shrinks by the factor $(\kappa \slash \lambda) \slash \gamma = (1-f) \slash (1-\gamma f)$. Whenever $f > 0$, then there is a strict decrease in the discount factor.
\end{proof}

\section{Hidden K-GLCPs of type $b$ are K-GLCPs of type $b + \mathbf{1}$} \label{sec:HidKKRelation}

We discuss a technique to formulate hidden K-GLCPs as K-GLCPs. The dimension is preserved whereas the size of each block increases by one. We allow ourselves to cheat a bit. The reduction requires a hidden K-witness of the input matrix, which is usually not known and eventually needs to be computed. This can be done through solving an LP (see remarks on p.~\pageref{page:CompHidKWit}). Nevertheless, the result suggests that the K-GLCP is as difficult as the hidden K-GLCP---as difficult as linear programming over grids. Such a result also follows from our previous observations (cf. Sections \ref{sec:HidKGLCPLP} and \ref{sec:DiscMDPs}). At first glance, this seems surprising because every simple principal pivoting method solves ordinary K-LCPs of order $n$ in at most $2n$ pivot steps \cite{FonFukGar:Pivoting}.

\begin{thm} \label{thm:RelHiddKK}
Let any hidden K-GLCP$(M,q)$ of type $b \in \mathbb{N}^n$ be given. Suppose that $(X,Y)$ is a hidden K-witness of $M$. An $N \subseteq N(b + \mathbf{1})$ is a solution basis to the hidden K-GLCP$(M,q)$ if and only if $N \cup \{(j,b_j + 2) : \, j \in [n]\}$ is a solution basis to the
$$\text{K-GLCP}([Y|X],[q - M f|-f])$$
for any $f \in \mathbb{R}^n > \mathbf{0}$. Moreover, the USO of the grid G$(b + \mathbf{1})$ arising from the hidden K-GLCP is fully contained in the USO of the grid G$(b + \mathbf{2})$ arising from the K-GLCP.
\end{thm}

\begin{proof}
We first prove that the USO $\phi$ of the grid G$(b + \mathbf{2})$ arising from the K-GLCP contains the USO $\varphi$ of the grid G$(b + \mathbf{1})$ arising from the hidden K-GLCP. In doing so, we assume nondegeneracy; otherwise we cannot speak of USOs. The feasible block vectors $x$ of type $b + \mathbf{2}$ to the system
	$$\left[ I | -[Y | X] \right] x = [q - M f|-f],$$
	where $I$ is supposed to be the horizontal block identity matrix of type $b +\mathbf{1}$, are determined by
	$$x(v):=[[Yv + q-M f | Xv - f] | v ]$$
  for $v \in \mathbb{R}^n$. For any maximal complementary $B \subseteq N(b + \mathbf{1})$, the orientation of the edges in $\phi$ that are incident to vertex $B$ are obtained through picking $u \in \mathbb{R}^n$ such that $x(u)_B=\mathbf{0}$. Note that the vector $u$ is uniquely determined because $[I|-[Y|X]]_{N}$ for $N:=N(b + \mathbf{2}) \backslash B$ is nonsingular. At the same time the subvectors
	$$x(v)_{N(b + \mathbf{1})}=[Yv + q-M f | Xv - f]$$
	for $v \in \mathbb{R}^n$ are the feasible vectors of type $b + \mathbf{1}$ to the system
	$$\left[ I | -M \right] x = q,$$
	where $I$ is supposed to be the horizontal block identity matrix of type $b$. The orientations of the edges in $\varphi$ that are incident to vertex $B$ are determined by the block vector $x(u)_{N(b + \mathbf{1})}$. Hence, the orientations of the edges in directions $N(b + \mathbf{1})$ coincide.
	
	Secondly, the last constraint in each block $j \in [n]$ of the K-GLCP under consideration is of the form
$$w^j_{b_j+1} - x^j z = -f_j,$$
where $x^j$ denotes the $j$th row in $X$. Since $X$ is a K-matrix, we have $\smash{x^j_j > 0}$ and $x^j_i \leq 0$ for all $i \neq j$. To satisfy the constraint with nonnegative $\smash{w^j_{b_j+1}}$ and $z$, we require that $z_j > 0$. Hence, every solution basis is of the form $N \cup \{(j,b_j + 2) : \, j \in [n]\}$ for some $N \subseteq N(b + \mathbf{1})$. Any solution vertex is contained in the subgrid USO $\varphi$ arising from the hidden K-GLCP.
\end{proof}

Note that for the first part of the proof it is not required that the vector $f$ is strictly positive---the assumption can be relaxed.

Next, we discuss the opposite direction with regard to degenerate instances.

\begin{lem} \label{lem:HidKmatrix}
 Let $M \in \mathbb{R}^{m \times n}$ be a K-matrix of type $b \in \mathbb{N}^n$. Any $M_{\overline{C}}M_C^{-1}$ for maximal complementary $C \subseteq N(b)$ is a hidden K-matrix of type $b - \mathbf{1}$.
\end{lem}

\begin{proof}
 The tuple $(M_C,M_{\overline{C}})$ is a hidden K-witness of $M_{\overline{C}}M_C^{-1}$.
\end{proof}

\begin{thm} \label{thm:KasHidK}
 Let any K-GLCP$(M,q)$ of type $b \in \mathbb{N}^n$ with $q_C \leq \mathbf{0}$ for $C = \{(j,b_j) : \, j \in [n]\}$ be given. If $N \subseteq N(b)$ is a solution basis to the
 $$\text{hidden K-GLCP}(M_{\overline{C}}M_C^{-1}, q_{\overline{C}}-M_{\overline{C}}M_C^{-1} q_C),$$
then $N \cup \{(j,b_j+1) : \, j \in [n] \}$ is a solution basis to the K-GLCP$(M,q)$.
\end{thm}

\begin{proof}
 By Lemma \ref{lem:HidKmatrix}, the tuple $(M_C,M_{\overline{C}})$ is a hidden K-witness of $M_{\overline{C}}M_C^{-1}$, and thus the GLCP of type $b - \mathbf{1}$ under consideration is a hidden K-GLCP. According to Theorem \ref{thm:RelHiddKK}, the statement is true for perturbed right-hand sides $q(\epsilon):=[q_{\overline{C}}|q_C-\boldsymbol{\epsilon}]$ for $\epsilon > 0$. Since the solution bases of a P-GLCP are stable for $\epsilon$ sufficiently close to $0$, the result follows.
\end{proof}

Actually, any K-GLCP$(M,q)$ of type $b \in \mathbb{N}^n$ with $q_C \leq \mathbf{0}$ for some maximal complementary $C \subset N(b)$ can be reduced to a hidden K-GLCP---a reordering of rows in the same block is not a problem.

\medskip

For discounted MDPs and policy iteration methods, the intermediate total discounted rewards in expectation are statewise monotonically nondecreasing. The property generalizes to the setting of linear complementarity. For K-GLCPs and simple principal pivot methods as well as Lemke's method, the intermediate $z_j$ for $j \in [n]$ are monotonically nondecreasing. This can also be seen as a generalization of the \emph{local uniformity} property satisfied by the K-USOs of the $n$-cube~\cite{FonFukGar:Pivoting}.

For a P-GLCP$(M,q)$ of type $b \in \mathbb{N}^n$ and any vertex $B$ of the grid G$(b + \mathbf{1})$, let $(w^B,z^B)$  denote the corresponding basic solution. Note that $[w^B|z^B]_{\overline{B}} = [I|-M]_{\overline{B}}^{-1}q$ and $[w^B|z^B]_B = \mathbf{0}$.

\begin{prop} \label{prop:MonIncZKGLCP}
 For a K-GLCP$(M,q)$ of type $b \in \mathbb{N}^n$, let $B$ and $C:=(B \backslash \{ (j,i)\}) \cup \{(j,k)\}$ be adjacent vertices of the grid G$(b + \mathbf{1})$. If $([I|-M]_{\overline{B}}^{-1}q)^j_k < 0$, then $z^C \geq z^B$ while $z^C \neq z^B$.
\end{prop}

\begin{proof}
Let $U:=C \cap N(b)$, $V:=\{j : \, (j, b_j+1) \in C \}$, and $\overline{V}:=\{j : \, (j, b_j+1) \in \overline{C} \}$. Note that $\smash{z^C_{\overline{V}}=-M_{U\overline{V}}^{-1}q_U}$. We have
$$z^C_{\overline{V}} - z^B_{\overline{V}} = -M_{U\overline{V}}^{-1}(q_{U}+M_{U\overline{V}}z^B_{\overline{V}}),$$
where we claim that the right-hand side is nonnegative. First, $-M_{U\overline{V}} ^{-1} \leq \mathbf{0}$ because $M_{U\overline{V}}$ is a square K-matrix. Secondly, it holds that
\begin{align*}
q_U+M_{U\overline{V}}z^B_{\overline{V}} & \leq q_{U}+ M_{U\overline{V}}z^B_{\overline{V}} + M_{UV}z^B_{V} \\
 & = q_{U}+ M_{U}z^B \leq \mathbf{0}.
\end{align*}
The first inequality holds because $M_{UV} \leq \mathbf{0}$ and $z^B_{V} \leq \mathbf{0}$. More precisely, if by any chance $k = b_j +1$ for some $j \in [n]$, then $j \in V$ and $\smash{z^B_j = ([I|-M]_{\overline{B}}^{-1}q)^j_k < 0}$ is the only negative entry in $\smash{z^B_{V}}$,  otherwise $\smash{z^B_{V}=\mathbf{0}}$. For the second inequality to hold, it is crucial that $\smash{([I|-M]_{\overline{B}}^{-1}q)^j_k < 0}$ if $k \in [b_j]$. Hence $\smash{z^C \geq z^B}$. Obviously $\smash{z_V^C \neq z_V^B}$ if $k = b_j +1$. Otherwise $\smash{z_{\overline{V}}^C \neq z_{\overline{V}}^B}$ because $\smash{M_{U\overline{V}} ^{-1}}$ is a P-matrix and thus has strictly positive diagonal elements.
\end{proof}

\section{A strongly polynomial reduction from the P-GLCP to the P-LCP} \label{sec:PRed}

Nohan, Neogy, and Sridhar \cite{mohan1996generalized} presented a reduction from the GLCP to the ordinary LCP. The reduction is such that P-GLCPs map to LCPs with singular Q-matrices. We would like to propose a specialized reduction that preserves P-property. In doing so, the computational complexity is guaranteed to stay in $\text{NP} \cap \text{coNP}$.



Consider a P-matrix $M \in \mathbb{R}^{m \times n}$ and a vector $q \in \mathbb{R}^{m}$, both of the same type $b \in \mathbb{N}^n$. We assume without loss of generality that $m^j_{\cdot j}=\mathbf{1}$ for every $j \in [n]$ (cf. Lemma \ref{lem:PGLCPScaling}).

The P-GLCP($M,q$) is to find a vector $z \in \mathbb{R}^n$ and a block vector $w \in \mathbb{R}^m$ of type $b$ such that
\begin{equation} \label{eq:origGLCP}
 \begin{aligned}
  z_j & = w^j_i - \sum_{k \neq j} m^j_{ik} z_k - q^j_i& \fall j \in [n], \fall i \in [b_j], \\
  w,z & \geq \mathbf{0}, & \\
  z_j \prod_{i = 1}^{b_j} w^j_i & = 0 & \fall j \in [n].
 \end{aligned}
\end{equation}

The reduction proceeds by means of iterations. In every iteration, the size of some block decreases by one. This is attained by creating two new blocks, each of size one.

Consider the following GLCP, where we assume that $b_n \geq 2$. Find a vector $v \in \mathbb{R}^{n+2}$ and block vector $u \in \mathbb{R}^{m+1}$ of type $(b_1, \ldots, b_{n-1}, b_n-1, 1, 1)$ such that

\begin{subequations} \label{eq:reducedGLCP}
 \begin{equation}
  \begin{aligned}
   v_j & = u^j_i - \sum_{k \neq j} m^j_{ik} v_k - m^j_{in} v_{n+2} - q^j_i & \fall j \in [n-1], \fall i \in [b_j], \\
	 v_n & = u^n_i - \sum_{k \neq n} m^n_{ik} v_k - q^n_i & \fall i \in [b_n-1], \\
	 v_{n+1} & = u^{n+1}_{1} - \sum_{k \neq n} m^n_{b_{n}k} v_k - q^n_{b_n}, & \\
	 v_{n+2} & = u^{n+2}_1 + v_{n+1} - v_n, &
	\end{aligned}
 \end{equation}
both $u$ and $v$ are nonnegative, and the complementarity constraints are satisfied, i.e., we have
 \begin{equation}
  \begin{aligned}
	 v_j \prod_{i = 1}^{b_j} u^j_i = 0 \quad \fall j \in [n-1], \quad v_n \prod_{i = 1}^{b_n-1} u^n_i = 0, \quad v_{n+1} u^{n+1}_1 = 0, \quad v_{n+2} u^{n+2}_1  = 0.
  \end{aligned}
 \end{equation}
\end{subequations}
Suppose that $(w, z)$ is the solution to the P-GLCP \eqref{eq:origGLCP}. The reader is asked to verify that a solution $(u,v)$ to the GLCP \eqref{eq:reducedGLCP} is obtained as follows:
\begin{equation*}
 \begin{aligned}
  v_j & :=z_j & \fall j \in [n-1], \\
	u^j & :=w^j & \fall j \in [n-1], \\
	v_n & := z_n - \min \left\{z_n, w^n_1, \ldots, w^n_{b_n-1} \right\}, \\
	u^n_i & := w^n_i - \min \left\{z_n, w^n_1, \ldots, w^n_{b_n-1} \right\} & \fall i \in [b_n-1], \\
	v_{n+1} & := z_n - \min \{z_n, w^n_{b_n}\}, \\
	u^{n+1}_1 & := w^n_{b_n} - \min \{z_n, w^n_{b_n}\}, \\
	v_{n+2} & := \begin{cases} 0 &\mbox{if } v_{n} >  v_{n+1}, \\ v_{n+1} - v_{n} & \mbox{otherwise, } \end{cases} \\
	u^{n+2}_1 & := \begin{cases} v_{n} - v_{n+1} &\mbox{if } v_{n} >  v_{n+1}, \\ 0 & \mbox{otherwise. } \end{cases}
 \end{aligned}
\end{equation*}
The other direction is crucial for the reduction to work.

\begin{lem} \label{lem:solRed}
For a solution $(u,v)$ to the GLCP \eqref{eq:reducedGLCP}, the solution $(w,z)$ to the P-GLCP \eqref{eq:origGLCP} is obtained as follows:
\begin{equation*}
 \begin{aligned}
  z_j & :=v_j & \fall j \in [n-1], \\
	w^j & :=u^j & \fall j \in [n-1], \\
	z_n & := \max \{v_n, v_{n+1} \}, \\
	w^n_i & := \begin{cases} u^n_i &\mbox{if } v_{n} >  v_{n+1}, \\ u^n_i + v_{n+2} & \mbox{otherwise, } \end{cases} & \fall i \in [b_n-1], \\
	w^n_{b_n} & := \begin{cases} u^{n+1}_1 + u_1^{n+2} &\mbox{if } v_{n} >  v_{n+1}, \\ u^{n+1}_1 & \mbox{otherwise. } \end{cases}
 \end{aligned}
\end{equation*}
\end{lem}

\begin{proof}
 \textit{Case} $v_n > v_{n+1}$. Then $z_n=v_n$. Since $v_{n+2}=0$, the first $n-1$ blocks of P-GLCP \eqref{eq:origGLCP} and the corresponding complementarity constraints are obviously satisfied. For rows $i \in [b_n-1]$ of block $n$, we verify that
\begin{align*}
 z_n & = v_n \\
     & = u^n_i - \sum_{k \neq n} m^n_{ik} v_k - q^n_i \\
		 & = w^n_{i} - \sum_{k \neq n} m^n_{ik} z_k - q^n_{i}.
\end{align*}
Similarly, for row $b_n$, we have
\begin{align*}
 z_n & =v_{n+1} + u_1^{n+2} \\
     & = u^{n+1}_{1} - \sum_{k \neq n} m^n_{b_{n}k} v_k - q^n_{b_n} + u_1^{n+2} \\
		 & = w^n_{b_n} - \sum_{k \neq n} m^n_{b_{n}k} z_k - q^n_{b_n}.
\end{align*}
The complementarity constraint holds because
 $$z_n \prod_{i = 1}^{b_n-1} w^n_i =  v_n \prod_{i = 1}^{b_n-1} u^n_i = 0.$$

 \textit{Case} $v_{n+1} \geq v_n$. Then $z_n=v_{n+1}$. Since $v_{n+2}=v_{n+1}-v_n$, the first $n-1$ blocks of GLCP \eqref{eq:origGLCP} and the corresponding complementarity constraints are satisfied. For rows $i \in [b_n-1]$ of block $n$, we verify that
 \begin{align*}
 z_n & = v_n + v_{n+2} \\
     & = u^n_i - \sum_{k \neq n} m^n_{ik} v_k - q^n_i + v_{n+2} \\
		 & = w^n_{i} - \sum_{k \neq n} m^n_{ik} z_k - q^n_{i}.
\end{align*}
Similarly, for row $b_n$, we have
 \begin{align*}
  z_n & = v_{n+1} \\
			& = u^{n+1}_{1} - \sum_{k \neq n} m^n_{b_{n}k} v_k - q^n_{b_n}  \\
			& = w^n_{b_n} - \sum_{k \neq n} m^n_{b_{n}k} z_k - q^n_{b_n}.
 \end{align*}
The complementarity constraint holds because either $z_n=v_{n+1}=0$ or $w^n_{b_n}=u^{n+1}_1=0$.
\end{proof}

As seen, there is a one-to-one correspondence between the solutions to the two GLCPs. From this it may already follow that \eqref{eq:reducedGLCP} is also a P-GLCP---because of arbitrary right-hand side $q$. Recall also (f) in Theorem \ref{thm:altBlockP}. We give a formal proof of this fact for the sake of completeness.

The reduction yields a GLCP($M',q'$) of type $b':=(b_1, \ldots, b_{n-1}, b_n-1, 1, 1)$, where
$$M'=\begin{bmatrix}
M^1_{\cdot [n-1]} & m^1_{\cdot n} & \mathbf{0} & m^1_{\cdot n} \\
M^2_{\cdot [n-1]} & m^2_{\cdot n} & \mathbf{0} & m^2_{\cdot n} \\
\vdots & \vdots & \vdots & \vdots \\
M^{n-1}_{\cdot [n-1]} & m^{n-1}_{\cdot n} & \mathbf{0} & m^{n-1}_{\cdot n} \\
M^n_{[b_n-1][n-1]} & \mathbf{1} & \mathbf{0} & 0 \\
m^n_{b_n [n-1]} & 0 & 1 & 0 \\
\mathbf{0}^T & 1 & -1 & 1
\end{bmatrix}
\text{ and } q'=
\begin{bmatrix} q^1 \\ q^2 \\ \vdots \\ q^{n-1} \\ q^n_{[b_n-1]} \\ q^n_{b_n} \\ 0 \end{bmatrix}.$$
The dimension of the problem is $n+2$ and the total number of block sizes equals $m+1$.

\begin{lem} \label{lem:isP}
 The matrix $M'$ is a P-matrix of type $b'$.
\end{lem}

\begin{proof}
 Let
 $$B':=\{(1,i), (2, j), \dots, (n-1,k), (n,l), (n+1, 1), (n+2,1)\}$$
 be any maximal complementary subset of $N(b')$. The corresponding representative submatrix $M'_{B'}$ is of the form
 $$\begin{bmatrix}
  A & c \\ d^T & 1
 \end{bmatrix}$$
 for $A:=(M'_{B'})_{[n+1] [n+1]}$, $c:=(m^1_{i n}, m^2_{j n}, \ldots, m^{n-1}_{k n} , 0, 0)^T$, and $d^T:=(\mathbf{0}^T, 1, -1)$. 
Note that $\det  M'_{B'} = \det A - cd^T$, where
 $$A - cd^T=
 \begin{bmatrix}
  m^1_{i [n-1]} & 0 & m^1_{i n} \\
  m^2_{j [n-1]} & 0 & m^2_{j n} \\
  \vdots & \vdots & \vdots \\
  m^{n-1}_{k [n-1]} & 0 & m^{n-1}_{k n} \\
  m^n_{l[n-1]} & 1 & 0  \\
  m^n_{b_n [n-1]} & 0 & 1
 \end{bmatrix}.
 $$
 Hence $\det M'_{B'} = \det M_B > 0$ for maximal complementary subset
 $$B:=(B' \backslash \{(n,l), (n+1,1), (n+2,1)\}) \cup \{(n,b_n)\}$$
 of $N(b)$. By similar arguments, every principal minor of $M'_{B'}$ is equal to some principal minor of $M_B$.
\end{proof}

\begin{thm} \label{lem:KGLCPRed}
 The P-GLCP reduces to the P-LCP in strongly polynomial time.
\end{thm}

\begin{proof}
 Consider any P-GLCP of type $b \in \mathbb{N}^n$. We apply the reduction presented above $m-n$ times until each block is of size one. By Lemma \ref{lem:isP}, we end up with an ordinary P-LCP of dimension $2m-n$. Given the unique solution to the final P-LCP, the unique solution to the initial P-GLCP is obtained trough retracing the reduction steps backwards using Lemma~\ref{lem:solRed}.
\end{proof}

The reduction is strongly polynomial. It certainly does not preserve the K-property. We keep investigating whether some modification preserves the hidden K-property.

This result illustrates that the dimension $n$ of a P-GLCP can be sacrificed in order to get an ordinary P-LCP. It would be interesting to investigate whether there is a connection to a result obtained by G\"artner, Morris, and R\"ust \cite{GarMorRus:Unique}, who observed that certain randomized path-following algorithms for general USOs of $n$-grids run in an expected linear number in $m$ of pivot steps for fixed dimension $n$.


\section{A reduction from the K-GLCP to the K-GLCP with blocks of size at most two} \label{sec:KReduction}

Discounted MDPs can be formulated as K-GLCPs (cf.~Proposition \ref{prop:MDPsAsKGLCPs}). Hence, existence of a strongly polynomial reduction from the K-GLCP to the ordinary K-LCP, which admits efficient pivoting methods \cite{FonFukGar:Pivoting}, would be a big surprise. We, nevertheless, propose a (strongly) polynomial reduction from the K-GLCP with blocks of arbitrary size to the K-GLCP with blocks of size at most two. This is probably best what we can expect. The reduction is slightly different from the one presented in the previous section and operates on stochastic K-GLCPs.

\begin{lem} \label{lem:NormFormComp2}
 Let $M \in \mathbb{R}^{m \times n}$ be a K-matrix. If either a vector $x \in \mathbb{R}^n > \mathbf{0}$ with $Mx>\mathbf{0}$ or a proper hidden K-witness of $M$ is known, then a stochastic form of $M$ can be computed in strongly polynomial time.
\end{lem}

\begin{proof}
 We apply the construction scheme outlined in the proof of Lemma \ref{lem:NormFormComp}. We pick $D:=\diag(x)$ or $D:=\diag(X \mathbf{1})$, where $(X,Y)$ is any proper hidden K-witness of $M$. It is crucial that $MD\mathbf{1}=MX\mathbf{1}=Y\mathbf{1} > \mathbf{0}$.
\end{proof}

Let $M \in \mathbb{R}^{m \times n}$ be a stochastic K-matrix of type $b \in \mathbb{N}^n$ and $q \in \mathbb{R}^m$. The matrix $M$ is of the form $E(b)-\gamma P$ for some rowstochastic $P \in \mathbb{R}^{m \times n}$ and $\gamma \in [0,1)$.

The stochastic K-GLCP($M,q$) is the problem to find a vector $z \in \mathbb{R}^n$ and a block vector $w \in \mathbb{R}^m$ of type $b$ such that

\begin{equation} \label{eq:origKGLCP}
 \begin{aligned}
  z_j & = w_i^j + \gamma \sum_{k=1}^n p_{ik}^j z_k - q_i^j & \fall j \in [n], \fall i \in [b_j], \\
	w,z & \geq \mathbf{0}, & \\
	z_j \prod_{i=1}^{b_j} w_i^j & = 0 & \fall j \in [n],
 \end{aligned}
\end{equation}

The reduction proceeds by means of iterations. In every iteration, the size of some block decreases by one, which is attained by creating a new block of size two.

Consider the following GLCP, where we assume that $b_n \geq 3$. Find a vector $v \in \mathbb{R}^{n+1}$ and block vector $u \in \mathbb{R}^{m+1}$ of type
$(b_1, \ldots, b_{n-1}, b_n-1, 2)$  such that

\begin{subequations} \label{eq:reducedKGLCP}
 \begin{equation}
  \begin{aligned}
   v_j & = u_i^j + \gamma \sum_{k=1}^{n-1} p_{ik}^j v_k + \gamma p_{in}^j v_{n+1} - q_i^j & \fall j \in [n-1], \fall i \in [b_j], \\
	 v_n & = u_i^n + \gamma \sum_{k=1}^{n-1} p_{ik}^n v_k + \gamma p_{in}^n v_{n+1}- q_i^n &  \fall i \in [b_n-1], \\
	 v_{n+1} & = u_1^{n+1} + \gamma \sum_{k=1}^{n-1} p_{b_nk}^n v_k + \gamma p_{b_n n}^n v_{n+1} - q_{b_n}^n, &  \\
	 v_{n+1} & = u_2^{n+1} + v_n,
	\end{aligned}
 \end{equation}
 both $u$ and $v$ are nonnegative, and the complementarity constraints are satisfied, i.e., we have
 \begin{equation}
  \begin{aligned}
	 v_j \prod_{i=1}^{b_j} u_i^j = 0 \quad \fall j \in [n-1], \quad v_n \prod_{i=1}^{b_n-1} u_i^n = 0, \quad v_{n+1} \prod_{i=1}^2 u_i^{n+1} = 0.
  \end{aligned}
 \end{equation}
\end{subequations}

Suppose that $(w, z)$ is a solution to the K-GLCP \eqref{eq:origKGLCP}. The reader is asked to verify that a solution $(u,v)$ to the GLCP \eqref{eq:reducedKGLCP} is obtained as follows:
\begin{equation*}
 \begin{aligned}
  v_j & :=z_j & \fall j \in [n-1], \\
	u^j & :=w^j & \fall j \in [n-1], \\
	v_n & := z_n - \min \left\{z_n, w^n_1, \ldots, w^n_{b_n-1} \right\}, &  \\
	u^n_i & := w^n_i - \min \left\{z_n, w^n_1, \ldots, w^n_{b_n-1} \right\} & \fall i \in [b_n-1], \\
	v_{n+1} & := z_n, \\ 
	u^{n+1}_1 & := w^n_{b_n}, & \\ 
	u^{n+1}_2 & := \min \left\{z_n, w^n_1, \ldots, w^n_{b_n-1} \right\}. &
 \end{aligned}
\end{equation*}

The crucial direction is proven formally.

\begin{lem} \label{lem:SolKRed}
  For a solution $(u,v)$ to the GLCP \eqref{eq:reducedKGLCP}, the solution $(w,z)$ to the K-GLCP \eqref{eq:origKGLCP} is obtained as follows:
	\begin{equation*}
   \begin{aligned}
    z_j & :=v_j & \fall j \in [n-1], \\
	  w^j & :=u^j & \fall j \in [n-1], \\
	  z_n & := v_{n+1}, \\
	  w^n_i & :=  u^n_i + u^{n+1}_2 & \fall i \in [b_j-1], \\
	  w^n_{b_n} & := u^{n+1}_1. &
   \end{aligned}
  \end{equation*}
\end{lem}

\begin{proof}
 Note that $z_n:=v_{n+1} \geq v_n$. The first $n-1$ blocks of K-GLCP \eqref{eq:origKGLCP} and the corresponding complementarity constraints are obviously satisfied. For rows $i \in [b_n-1]$ of block $n$, we remark that
 \begin{align*}
  z_n & = v_{n+1} \\
	    & = u^{n+1}_2 + v_n  \\
			& = u^{n+1}_2 + u^n_i + \gamma \sum_{k=1}^{n-1} p_{ik}^n v_k + \gamma p_{in}^n v_{n+1}- q_i^n  \\
			& = w^n_i + \gamma \sum_{k=1}^n p_{ik}^n z_k - q_i^n.
 \end{align*}
 Similarly, for row $b_n$, we have
 \begin{align*}
  z_n & = v_{n+1} \\
	    & = u_1^{n+1} + \gamma \sum_{k=1}^{n-1} p_{b_nk}^n v_k + \gamma p_{b_n n}^n v_{n+1} - q_{b_n}^n, \\
			& = w^n_{b_n} + \gamma \sum_{k=1}^{n} p_{ik}^n z_k - q_{b_n}^n.
 \end{align*}
The complementary condition is satisfied as well. If $v_{n+1}=0$ or $u^{n+1}_1=0$, then $z_n = 0$ or $w^n_{b_n}=0$, respectively. Otherwise, we have $u^{n+1}_2=0$. Since $z_n=v_{n+1}=v_n$ and $w_i^n=u_i^n$ for all $i \in [b_j-1]$, it follows that
$$z_n \prod_{i=1}^{b_n-1} w_i^n = v_n \prod_{i=1}^{b_n-1} u_i^n = 0.$$
\end{proof}

The reduction yields a GLCP($M',q'$) of type $b':=(b_1, \ldots, b_{n-1}, b_n-1, 2)$, where
$$M'=\begin{bmatrix}
M^1_{\cdot [n-1]} & \mathbf{0} & m^1_{\cdot n} \\
M^2_{\cdot [n-1]} & \mathbf{0} & m^2_{\cdot n} \\
\vdots & \vdots & \vdots & \vdots \\
M^{n-1}_{\cdot [n-1]} & \mathbf{0} & m^{n-1}_{\cdot n} \\
M^n_{[b_n-1][n-1]} & \mathbf{1} & - \gamma p^n_{\cdot n}  \\
m^n_{b_n [n-1]} & 0 & m^n_{b_n n} \\
\mathbf{0}^T & -1 & 1
\end{bmatrix}
\text{ and } q'=
\begin{bmatrix} q^1 \\ q^2 \\ \vdots \\ q^{n-1} \\ q^n_{[b_n-1]} \\ q^n_{b_n} \\ 0 \end{bmatrix}.$$
The dimension of the problem is $n+1$ and the total number of block sizes equals $m+1$.

\begin{lem}
 The matrix $M'$ is a K-matrix of type $b'$.
\end{lem}

\begin{proof}
 The matrix $M'$ is obviously a Z-matrix. It is left to prove it satisfies the P-property. First, consider any maximal complementary subset $B'$
 of $N(b')$ that contains $(n+1,1)$. Let $l \in [b_n]$ be such that $(n,l) \in B'$. Then $\det M'_{B'} = \det M_B > 0$ for maximal complementary subset
 $$B:=(B' \backslash \{(n,l), (n+1,1)\}) \cup \{(n,b_n)\}$$
 of $N(b)$. Similar arguments apply for the principal minors of $M'_{B'}$.

 Secondly, consider any such $B'$ that contains $(n+1,2)$. We think of the corresponding representative submatrix $M'_{B'}$ as being of the form
 $$\begin{bmatrix}
  A & c \\ d^T & 1
 \end{bmatrix},$$
 for some well-defined $n \times 1$ vectors $c$ and $d$ and $n \times n$ matrix $A$. Recall that $\det  M'_{B'} = \det A - cd^T$. By verifying that $A - cd^T=M_B$ for $B:=B' \backslash \{(n+1,2)\}$, it follows that $\det M'_{B'} = \det M_B > 0$. By similar arguments, the principal minors of $M'_{B'}$ are positive, too.
\end{proof}

In order to iteratively execute reduction steps, a stochastic form of $M'$ has to be computed.

\begin{lem} \label{lem:PropFormAgain}
 The matrix $LM'H$ for
 \begin{align*}
	 L & :=\diag(1 \slash 2, \ldots, 1 \slash 2, 1, 1 \slash 2, 1) \in \mathbb{R}^{(m+1) \times (m+1)} \text{ and } \\
	 H & :=\diag(1, \ldots, 1, (1 + \gamma) \slash 2, 1) \in \mathbb{R}^{(n+1) \times (n+1)}
 \end{align*}
 is a stochastic K-matrix of type $b'$ with factor $(1+\gamma) \slash 2$.
\end{lem}

\begin{proof}
 Note that $M \mathbf{1} = (1- \gamma) \mathbf{1}$. Thus $LM'H \mathbf{1} = ((1 - \gamma) \slash 2) \mathbf{1}$. The largest diagonal element in any representative submatrix of $LM'H$ equals $1$. Therefore, each representative submatrix of $LM'H$ can be represented as $I - T$ for some $T \geq \mathbf{0}$. Since $(I - T) \mathbf{1} = \mathbf{1} - T\mathbf{1} = ((1 - \gamma) \slash 2) \mathbf{1}$, it follows that $T \mathbf{1} = ((1 + \gamma) \slash 2) \mathbf{1}$. Let $S:=(2 \slash (1 + \gamma)) T$. Then obviously $S \geq \mathbf{0}$ with $S \mathbf{1}=\mathbf{1}$. The representative submatrix is represented as $I - ((1 + \gamma) \slash 2) S$. Hence $LM'H$ is a stochastic K-matrix with factor $(1 + \gamma) \slash 2 \in [0,1)$.
\end{proof}

Note that the discount factor increases by $(1- \gamma) \slash 2$. It would be interesting to investigate whether the reduction can be modified such that the increase will be less. The discount factor is an indicator for the difficulty of stochastic K-GLCPs. Intuitively, for small $\gamma$ the representative submatrices are almost identity matrices---pivoting methods converge fast. For $\gamma$ close to $1$, on the other hand, the problem can be more difficult.
\begin{thm} \label{thm:RedKGLCP2KGLCP}
 The K-GLCP reduces to the K-GLCP with blocks of size at most two in polynomial time. The reduction is strongly polynomial for stochastic K-GLCPs.
\end{thm}

\begin{proof}
 If the K-GLCP instance is in stochastic form, we repeatedly apply the reduction presented above until each block is of size at most two. Otherwise, we first compute a stochastic form using Lemma \ref{lem:NormFormComp2} together with any given or computed proper hidden K-witness. In each reduction step, the current K-GLCP can be converted into stochastic K-GLCP through applying Lemma~\ref{lem:PropFormAgain} together with Lemma~\ref{lem:PGLCPScaling}. Finally, the solution to the initial problem is obtained through retracing the reduction steps backwards using Lemma~\ref{lem:SolKRed}.
\end{proof}

The following will become useful later.

\begin{lem} \label{lem:RHSNonPos}
 For K-GLCPs($M,q$) with $q_C \leq \mathbf{0}$ for some maximal complementary $C$, the reduction can be executed in such a way that for the final K-GLCP($M',q')$, we have $q'_{C'} \leq \mathbf{0}$ for some maximal complementary $C'$.
\end{lem}

\begin{proof}
 In every newly created block of size two, the right-hand side of the second constraint equals $0$. By making sure that to first split apart constraints with positive right-hand side, the result follows.
\end{proof}


This reduction scheme can be transformed into a scheme to reduce discounted MDPs to discounted MDPs whose each state has at most two actions assigned. See also Section~\ref{sec:HidKReduction} and Section~\ref{sec:Red2PG}, where the reduction from above is generalized to the setting of two adversary players.

\section{A reduction from Grid-LPs to Cube-LPs} \label{sec:HidKReduction}

We first discuss a reduction scheme for the hidden K-GLCP, the dual problem of linear programming over grids.

\begin{thm} \label{prop:HidKRed}
 Every hidden K-GLCP$(M,q)$ reduces to an ordinary hidden K-LCP in polynomial time. The reduction is strongly polynomial if a proper hidden K-witness of $M$ is known.
\end{thm}

\begin{proof}
  By Theorem \ref{thm:RelHiddKK}, every hidden K-GLCP$(M,q)$ admits a formulation as a K-GLCP$([Y|X],[q-M \mathbf{1}|-\mathbf{1}])$, where $(X,Y)$ is any hidden K-witness of $M$. Here, we actually ask for a proper witness, which eventually has to be computed. For convenience, we assume that a proper witness $(X,Y)$ is known whose binary encoding length is polynomial. Since $[Y|X] \mathbf{1} > \mathbf{0}$, the tuple $(I,[Y|X])$ is obviously a proper hidden K-witness of $[Y|X]$, and, by the construction scheme outlined in Lemma \ref{lem:NormFormComp2}, a combinatorially equivalent stochastic K-GLCP is immediately obtained. (No column scaling is required.) By Theorem \ref{thm:RedKGLCP2KGLCP}, the problem reduces to a K-GLCP with blocks of size at most two. Considering Lemma \ref{lem:RHSNonPos}, the reduction can be executed in such a way that some representative subvector of the right-hand side in the final K-GLCP will be nonpositive. Finally, Theorem \ref{thm:KasHidK} yields a hidden K-LCP.
\end{proof}

\begin{thm} \label{thm:GridLPCubeLPRed}
 Grid-LPs reduce to Cube-LPs in polynomial time. The reduction is strongly polynomial for Grid-LPs in stochastic form.
\end{thm}

\begin{proof}
 Theorem \ref{thm:LPUSOisMDPUSO} states that every Grid-LP has a stochastic form, where we eventually require to compute a proper hidden K-witness. Every Grid-LP in stochastic form represents a discounted MDP, which in turn can be formulated as a stochastic K-GLCP$(M,q)$ with $q < \mathbf{0}$. As described in the proof of Theorem \ref{prop:HidKRed}, the problem reduces to an ordinary hidden K-LCP, which is dual to some Cube-LP. A valid right-hand side $p$ for the Cube-LP is directly obtained because a (proper) witness of the square hidden K-matrix will be known through the kind of the reduction scheme.
\end{proof}

\begin{cor}
 Discounted MDPs reduce to their binary counterparts in strongly polynomial time.
\end{cor}

\begin{proof}
 Discounted MDPs are Grid-LPs in stochastic form and thus, by Theorem~\ref{thm:GridLPCubeLPRed}, reduce to \textsc{Cube-LPs} in strongly polynomial time. A proper witness $(X,Y)$ of the final square hidden K-matrix is known. Moreover, the matrix $[Y|X]$ is a stochastic K-matrix. Hence, a stochastic form of the Cube-LP is obtained with ease.
\end{proof}


Hence, every Grid-LP reduces to a Cube-LP in strongly polynomial time, at least in terms of combinatorics. It is an open question whether the reduction is strongly polynomial for Grid-LPs in arbitrary representation. Such a question boils down to the question as to whether there exists a strongly polynomial algorithm for the computation of proper hidden K-witnesses. It might be possible that Grid-LPs in stochastic form contain valuable information that can be exploited by some `smart' algorithm. The additional information is certainly of an algebraic kind---the representation does not influence the behavior of combinatorial simplex-type methods.

It would also be interesting to prove statements of the kind: ``if some specific pivot rule is inefficient for general Grid-LPs, then the rule is also inefficient for Cube-LPs". To approach such questions, we first have to interpret the proposed reductions in the context of USOs.

\section{A characterization of two-player stochastic games in terms of unique-sink orientations} \label{sec:TwoPlayerGames}

 We study \emph{two-player stochastic games} with perfect information and their relation to the GLCP. This family of games comprises several variants, which are  in the literature known as \emph{stochastic parity games}, \emph{stochastic mean-payoff games}, \emph{discounted stochastic games}, and \emph{simple stochastic games}. These variants are polynomially equivalent to each other~\cite{andersson2009complexity}. The games are in NP $\cap$ coNP, but their computational complexity is still open, even in the case of deterministic games. If the discount factor is supposed to be a constant, then the \emph{strategy iteration} algorithm is strongly polynomial~\cite{hansen2013strategy}. No polynomial-time solving method in the technical sense is known.

In the following, we shall discuss formulations as GLCPs. It has been known that these two-player stochastic games admit formulations as P-GLCPs~\cite{gartner2005simple, svensson2007linear, jurdzinski2008simple}. Here, we identify the subclass of P-GLCPs whose members represent games and also provide a characterization in terms of the combinatorial model of USOs.

A two-player stochastic game with perfect information is a stochastic process with discrete time steps. At each step, the process is in some state $j$, which is under control of exactly one of the two players. The player in charge then decides on an available action $i$. The process randomly moves to another state. The probability $\smash{p^j_{ik}} \in [0,1]$ that the process moves into state $k$ depends on the state $j$ and chosen action $i$. A moving step has reward $r^j_i \in \mathbb{R}$ assigned, which is likewise determined by the current state $j$ and action $i$.  We consider games over an infinite-time horizon, where the rewards are discounted by some factor $\gamma \in [0,1)$. To summarize, the following notations are used:
		
		\vspace{+0.3cm}
		\begin{tabular}{p{0.0cm}p{1cm}p{24cm}}
	   & $n$ & number of states, \\
		 & $S_{\textsc{max}}$ & states controlled by the \textsc{max} player ($S_{\textsc{max}} \subseteq [n]$), \\
		 & $S_{\textsc{min}}$ & states controlled by the \textsc{\textsc{min}} player ($S_{\textsc{min}} = [n] \backslash S_{\textsc{max}}$), \\
		 & $a_j$ & number of actions available in state $j \in [n]$, \\
		 & $r^j_i$ & reward for taking action $i \in [a_j]$ in state $j \in [n]$, \\
		 & $p^j_{ik}$ & conditional probability to arrive in state $k$ for state $j$ and action $i \in [a_j]$, \\
		 & $\gamma$ & discount factor ($\gamma \in [0,1)$).\\
	  \end{tabular}
		\vspace{+0.3cm}
		
The \textsc{max} player's aim is to maximize the total discounted reward in expectation, whereas the \textsc{min} player takes the role of an adversary player who wants to minimize the overall reward. A~\emph{policy} is a function $\pi$ that specifies for each state an action to take. The problem is then to find an \emph{optimal policy}, which is a policy such that none of the players is willing to switch to another action for any state he controls.

Let $v^{\pi}_j$ denote the total discounted reward in expectation for initial state $j$ when applying policy $\pi$. The equations
\begin{subequations} \label{eq:2PlayerGame}
 \begin{align}
  v^{\pi}_j & = \max_{i \in [a_j]} \left\{ r_i^j + \gamma \sum_{k=1}^n p_{ik}^j v^{\pi}_k \right\} \qquad \fall j \in S_{\textsc{max}}, \label{eq:2PG1}\\
  v^{\pi}_j & = \min_{i \in [a_j]} \left\{ r_i^j + \gamma \sum_{k=1}^n p_{ik}^j v^{\pi}_k \right\} \qquad \fall j \in S_{\textsc{min}}, \label{eq:2PG2}
 \end{align}
\end{subequations}
describe an optimality criterion for a policy $\pi$. The system \eqref{eq:2PlayerGame} has a unique solution~\cite{shapley1953stochastic}.

In order to formulate the game as a GLCP, we first transform the optimality criterion into a complementarity problem. The problem is to find a vector $v \in \mathbb{R}^n$ and a slack variable vector $u \in \mathbb{R}^m$ of type $a:=(a_j)_{j \in [n]}$ such that
\begin{equation} \label{eq:GCLPTwoPlayer}
 \begin{aligned}
  v_j & = u_i^j + r_i^j + \gamma \sum_{k=1}^n p_{ik}^{j} v_k & \fall j \in S_{\textsc{max}}, \fall i \in [a_j], \\
	v_j & = - u_i^j + r_i^j + \gamma \sum_{k=1}^n p_{ik}^{j} v_k & \fall j \in S_{\textsc{min}}, \fall i \in [a_j], \\
	u & \geq \mathbf{0}, & \\
	\prod_{i = 1}^{a_j} u_i^j & = 0 & \fall j \in [n].
 \end{aligned}
\end{equation}
The complementarity problem is not yet in proper form. We apply the procedure proposed by Jurdzi\'nski and Savani, who gave a P-LCP formulation for binary discounted games~\cite{jurdzinski2008simple}. We basically split apart the last row in each block.

Let $P:=(p^j_{ik})_{j,k \in [n], i \in [a_j]}$ and $r:=(r^j_i)_{j \in [n], i \in [a_j]}$, which are both of type $a$. Let $C:=\{(j,a_j) : \, j \in [n]\}$. Moreover, let $S$ be the $n \times n$ signature matrix with $s_{jj}=1$ for $j \in S_{\textsc{max}}$ and $s_{jj}=-1$ otherwise. Accordingly, let $\mathbf{S}$ denote a block diagonal matrix with $n$ blocks whose $j$th block is equal to the matrix $s_{jj} I$. The dimension of the blocks depends on the context. Replace variable vector $u$ with $[w|z]$, where $w$ is of type $a - \mathbf{1}$ and $z$ has $n$ entries. Let $E:=E(a-\mathbf{1})$. The problem \eqref{eq:GCLPTwoPlayer} can be written as
\begin{equation*}
 \begin{aligned}
  \mathbf{S}[E- \gamma P_{\overline{C}} | I - \gamma P_C] v & = [w | z] + \mathbf{S}[r_{\overline{C}}|r_C], \\
	w,z & \geq \mathbf{0}, & \\
	z_j \prod_{i = 1}^{a_j-1} w_i^j & = 0 & \fall j \in [n].
 \end{aligned}
\end{equation*}

Next, we would like to eliminate $v$. Since $I - \gamma P_C$ is a K-matrix and thus nonsingular, it follows from $S(I - \gamma P_C)v = z + S r_C$ that $v = (I - \gamma P_C)^{-1}S(z + Sr_C)$. By replacement, the equation system is of the form
\begin{align*}
 \mathbf{S}(E - \gamma P_{\overline{C}})v & = \mathbf{S}(E - \gamma P_{\overline{C}})(I - \gamma P_C)^{-1}S(z + Sr_C) \\
                  & = w + \mathbf{S}r_{\overline{C}}.
\end{align*}
Through a basic reordering of terms, we observe that
$$w - \mathbf{S}(E - \gamma P_{\overline{C}})(I - \gamma P_C)^{-1}Sz = \mathbf{S}(E - \gamma P_{\overline{C}})(I - \gamma P_C)^{-1}r_C - \mathbf{S} r_{\overline{C}}.$$

To summarize, we end up with the GLCP$(\mathbf{S}MS,\mathbf{S}q)$, where $M:=(E - \gamma P_{\overline{C}})(I - \gamma P_C)^{-1}$ and $q:=Mr_C - r_{\overline{C}}$. Note that $M$ is a hidden K-matrix of type $a - \mathbf{1}$. The matrix $\mathbf{S}MS$ is a P-matrix.

\begin{thm} Every two-player stochastic game of type $a \in \mathbb{N}^n$ admits a formulation as a \linebreak P-GLCP$(\mathbf{S}MS,\mathbf{S}q)$, where $M \in \mathbb{R}^{(m-n) \times n}$ is a hidden K-matrix of type $a - \mathbf{1}$, matrix $S$ is an $n \times n$ signature matrix, and $q$ is some vector in $\mathbb{R}^{m-n}$.
\end{thm}

Next, we discuss the opposite direction.

\begin{prop}
 Every P-GLCP$(\mathbf{S}MS,\mathbf{S}q)$ with a hidden K-matrix $M \in \mathbb{R}^{m \times n}$ of type $a \in \mathbb{N}^n$, an $n \times n$ signature matrix $S$, and a vector $q \in \mathbb{R}^m$ describes in terms of combinatorics a two-player stochastic game of type $a + \mathbf{1}$.
\end{prop}

\begin{proof}
 Let $(X,Y)$ be any proper hidden K-witness of $M$. We consider the P-GLCP$(\mathbf{S}LMHS,\mathbf{S}Lq)$ instead, where $L$ and $H$ are selected as in Lemma \ref{cor:StochFormHidK}. The induced USO stays the same. Let $P$ be rowstochastic such that $E(a + \mathbf{1}) -\gamma P=[LY|H^{-1}X]$ for some $\gamma \in [0,1)$. Let $C:=\{(j,a_j + 1) : \, j \in [n]\}$. Then
\begin{align*}
 \mathbf{S}LMDS & = \mathbf{S}(LY)(H^{-1}X)^{-1}S \\
       & = \mathbf{S}(E(a)-\gamma P_{\overline{C}})(I - \gamma P_C)^{-1} S.
\end{align*}
Pick any $r \in \mathbb{R}^{m+n}$ such that
$$Lq = (E(a)-\gamma P_{\overline{C}})(I - \gamma P_C)^{-1} r_C - r_{\overline{C}}.$$
The matrix $P$ and vector $r$ together encode the transition probabilities and rewards, respectively, of a two-player stochastic game of type $a + \mathbf{1}$.
\end{proof}

We conclude that there is a correspondence between two-player stochastic games and this specific subclass P-GLCPs arising from hidden K-matrices.

For single-player games, the signature matrix $S$ is either $I$ or $-I$. As previously reported, the single-player variants are equivalent to the hidden K-GLCPs, which in turn are characterized by the collection of LP-USOs of grids. A characterization of the two-player games in terms of USOs follows directly.

For a USO $\phi$ of an $n$-grid, let $\phi^F$ for $F \subseteq [n]$ denote the USO obtained from $\phi$ by reversing all edges in directions $j \in F$.

\begin{prop}
 The two-player stochastic games with $n$ states are characterized by the collection of USOs $\phi^F$, where $\phi$ is an LP-USO of an $n$-grid and $F \subseteq [n]$.
\end{prop}

Every P-USO of the $3$-cube arises from some two-player binary stochastic game, which follows from the enumeration of the P-USOs~\cite{StiWat:Digraph-models} and the LP-USOs~\cite{bernd1998abstract} of the $3$-cube, respectively. Hence, some USOs arising from stochastic games contain directed cycles---strategy improvement algorithms are finite though. In general, P-USOs that do not arise from any game may exist. At first glance, such a conjecture is supported by the fact that two-player stochastic games can be solved in expected subexponential time~\cite{halman2007simple,svensson2007linear}---no such an algorithm is known for general P-GLCPs. However, these algorithms exploit the fact that $F$ (or signature matrix $S$) is known. Hence, such a reasoning is not exactly supportive. The P-matrices arising from games are \emph{hidden row diagonally dominant}~\cite{tsatso:generating}; and thus, they build a proper subclass of general P-matrices~\cite{MorNam:sandwiches}. Below we provide an alternative complementarity formulation that makes such an observation self-evident. 

The characterization of stochastic games in terms of USOs reveals another interesting fact. Simplex-type methods for discounted MDPs compute an optimal policy through obtaining solutions to related two-player stochastic games. We restrict the discussion to the binary case.

\begin{prop}
 Consider a binary discounted MDP$(P,r,\gamma)$ with $n$ states, and denote the arising LP-USO of the $n$-cube by $\phi$. For every $F \subseteq [n]$, determining the unique vertex in $\phi$ with outgoing edges in directions $F$ is polynomially equivalent to solving the two-player stochastic game $( S_{\textsc{max}}:=[n] \backslash F, S_{\textsc{min}}:=F,P,r,\gamma)$.
\end{prop}

\begin{proof}
 Due to combinatorial equivalence, the USO $\phi$ also arises from some hidden K-LCP$(M,q)$. In every pivot step, we arrive at some vertex of the $n$-cube. Let $F \subseteq [n]$ denote the directions of the vertex's outgoing edges. The vertex is the unique sink of the USO $\phi^F$ and therefore the solution to the two-player binary stochastic game described by the P-LCP$(SMS,Sq)$, where $S$ is the $n \times n$ signature matrix with $S_{jj}=1$ for $j \notin F$ and $S_{jj}=-1$ otherwise.
\end{proof}

For hidden K-LCPs$(M,q)$ a proper hidden K-witness of $M$ is usually not known. Thus, determining some specific vertex in the arising LP-USO is even more difficult but is still polynomially equivalent to obtaining the solution to some related two-player binary stochastic game.



\subsubsection*{An alternative complementarity formulation}
We may apply the reduction scheme proposed by Svensson and Vorobyov~\cite{svensson2007linear} instead.

Consider again the problem \eqref{eq:GCLPTwoPlayer}. In order to lower and upper bound the entries of the solution $v^*$, we let
$$h:=\sum_{i=0}^{\infty}\gamma^id=\frac{d}{1-\gamma} \quad \mbox{ for } d > \max \{ |r_i^j| : j \in [n] \text{ and } i \in [a_j]\}.$$
For a pair $(u,v)$ satisfying~\eqref{eq:GCLPTwoPlayer}, we must have $-h \leq v_j \leq h$ for each $j \in [n]$. Hence, the complementarity problem can be written in proper form by substituting $z_j - h$ for $v_j$ if $j \in S_{\textsc{max}}$ and $-z_j + h$ for $v_j$ otherwise, for each $j \in [n]$. We can then add the constraint $z \geq \mathbf{0}$. For convenience, we also substitute $w$ for $u$. Then, the problem is to find a vector $z \in \mathbb{R}^n$ and block vector $w \in \mathbb{R}^m$ of type $a$ such that
\begin{equation*} 
 \begin{aligned}
  w_i^j - z_j  + \gamma \left( \smashoperator[r]{\sum_{k \in S_{\textsc{max}}}} p_{ik}^{j} z_k - \smashoperator[r]{\sum_{k \in S_{\textsc{min}}}} p_{ik}^{j} z_k  \right) & = -r_i^j + \gamma \left( \smashoperator[r]{\sum_{k \in S_{\textsc{max}}}} p_{ik}^{j} - \smashoperator[r]{\sum_{k \in S_{\textsc{min}}}} p_{ik}^{j} \right) h - h  & \fall j \in S_{\textsc{max}}, \fall i \in [a_j], \\
  w_i^j - z_j - \gamma \left( \smashoperator[r]{\sum_{k \in S_{\textsc{max}}}} p_{ik}^{j} z_k - \smashoperator[r]{\sum_{k \in S_{\textsc{min}}}} p_{ik}^{j} z_k\right) & = \phantom{-}r_i^j - \gamma \left( \smashoperator[r]{\sum_{k \in S_{\textsc{max}}}} p_{ik}^{j} - \smashoperator[r]{\sum_{k \in S_{\textsc{min}}}} p_{ik}^{j}  \right) h - h  & \fall j \in S_{\textsc{min}}, \fall i \in [a_j],
 \end{aligned}
\end{equation*}
both $z$ and $w$ are nonnegative, and $z_j \prod_{i = 1}^{a_j} w_i^j = 0$ for every $j \in [n]$.

Let $S$ be the $n \times n$ signature matrix with $s_{jj}=1$ if $j \in S_{\textsc{max}}$ and  $s_{jj}=-1$ otherwise. The reduction results in a P-GLCP$(\mathbf{S}MS,\mathbf{S}q)$ with a stochastic K-matrix $M:=E(a)-\gamma P$ and $q:=-r+\gamma P S \mathbf{h} - \mathbf{S} \mathbf{h}$.

The arising matrices $\mathbf{S}MS$ are row diagonally dominant and also belong to the class of \emph{H-matrices}, which by definition are the matrices that are up to the signs of their components K-matrices. For square matrices, the intersection of P- and H-matrices is properly contained in the collection of hidden K-matrices~\cite{Pan:Hidden}.  Such a result cannot hold for general block matrices---USOs arising from the alternative complementarity formulation for games may contain directed cycles\footnote{In order to actually prove such a statement, we would have to formulate a generalization of Theorem~\ref{thm:RelHiddKK}. The result then follows from the fact that some game induces the P-USO of the $3$-cube that contains a directed cycle (USO no.~19 in~\cite{StiWat:Digraph-models}).} but hidden K-GLCPs induce LP-USOs.

\section{A strongly polynomial reduction from two-player stochastic games to their binary counterparts} \label{sec:Red2PG}

It is worth to investigate whether the complementarity formulations for two-player stochastic games from the last section can be reduced to ordinary LCPs in such a way that the game property is preserved. The reduction presented in Section~\ref{sec:PRed} obviously provides a reduction to ordinary P-LCPs, but the game property is probably lost. Fortunately, there is a generalization of the reduction scheme for the single-player games presented in Section~\ref{sec:KReduction}.

Consider any two-player stochastic game as defined in the previous section. The problem is to solve the optimality criterion \eqref{eq:2PlayerGame} for $v \in \mathbb{R}^n$, whose entries represent the total discounted reward in expectation for the $n$ states.

For reasons of simplicity, we directly reduce the optimality criterion. By ignoring formulations as GLCPs, we avoid unnecessary complexity. The reduction again proceeds by means of iterations. In every iteration, the number of actions of some state decreases by one, which is attained through creating some new state with two actions.

Let $\delta := (1 + \gamma) \slash 2$, and note that $\delta \in [0,1)$. Suppose that state $n$ is under control of the \textsc{min} player. Consider the following equation system, where $z \in \mathbb{R}^{n+1}$ is a variable vector.

\begin{subequations} \label{eq:Red2PlayerGame}
 {\normalsize
  \begin{align}
   z_j & = \max_{i \in [a_j]} \left\{ \frac{r_i^j}{2} + \delta \sum_{k = 1}^{n-1} \frac{\gamma}{2 \delta} p_{ik}^j z_k + \delta \frac{1}{2 \delta} z_j + \delta \frac{\gamma}{2 \delta} p^j_{in} z_{n+1} \right\} \qquad \fall j \in S_{\textsc{max}} , \label{eq:Red2PG1}\\
   z_j & = \min_{i \in [a_j]} \left\{ \frac{r_i^j}{2} + \delta \sum_{k = 1}^{n-1} \frac{\gamma}{2 \delta} p_{ik}^j z_k + \delta \frac{\gamma}{2 \delta} z_j + \delta \frac{\gamma}{2 \delta} p^j_{in} z_{n+1} \right\} \qquad \fall j \in S_{\textsc{min}} \backslash \{n\}, \label{eq:Red2PG2}\\
	 z_n & = \min_{i \in [a_n-1]} \left\{ r_i^n + \delta \sum_{k = 1}^{n-1} \frac{\gamma}{\delta} p_{ik}^n z_k + \delta \left(\frac{1}{\delta} -1\right) z_n + \delta \frac{\gamma}{\delta} p^n_{in} z_{n+1} \right\}, \label{eq:Red2PG3} \\
	 z_{n+1} & = \min \left\{ \frac{r_{a_n}^n}{2} + \delta \sum_{k = 1}^{n-1} \frac{\gamma}{2 \delta} p_{a_n k}^n z_k + \delta \left( \frac{\gamma}{2 \delta} p^n_{a_n n} + \frac{1}{2 \delta} \right) z_{n+1}, \delta z_n \right\}. \label{eq:Red2PG4}
	 \end{align}}
	\end{subequations}
	
We will prove that the equation system \eqref{eq:Red2PlayerGame} defines the optimality criterion of a two-player stochastic game of type $(a_1, \ldots, a_{n-1}, a_n-1, 2)$ with discount factor $\delta$, where the $(n+1)$th state is controlled by the \textsc{min} player. The ownership of the other states does not change. Moreover, the game values coincide with the game values of the original game. The latter fact is proven first.
	
\begin{lem} \label{lem:Red2PGGameVal}
 For a solution $z \in \mathbb{R}^{n+1}$ to \eqref{eq:Red2PlayerGame}, the solution $v \in \mathbb{R}^n$ to \eqref{eq:2PlayerGame} is given by $v_j := z_j$ for $j \in [n-1]$ and $v_n := z_{n+1}$.
\end{lem}

\begin{proof}
For the states $j \in S_{\textsc{max}}$, it follows from the equations \eqref{eq:Red2PG1} that
 \begin{equation*}
  \begin{aligned}
   2 v_j & = 2 z_j \\
       & = 2 \max_{i \in [a_j]} \left\{ \frac{r_i^j}{2} + \delta \sum_{k = 1}^{n-1}\frac{\gamma}{2 \delta} p_{ik}^j z_k + \delta \frac{1}{2 \delta} z_j + \delta \frac{\gamma}{2 \delta} p^j_{in} z_{n+1} \right\} \\
		   & = \max_{i \in [a_j]} \left\{ r_i^j + \gamma \sum_{k = 1}^{n-1} p_{ik}^j z_k + z_j + \gamma p^j_{in} z_{n+1} \right\} \\
			 & = \max_{i \in [a_j]} \left\{ r_i^j + \gamma \sum_{k = 1}^n p_{ik}^j v_k \right\} + v_j.
  \end{aligned}
 \end{equation*}
Thus, the equations \eqref{eq:2PG1} are satisfied by $v$. An analogous argumentation applies for the states $j \in S_{\textsc{min}} \backslash \{n\}$. Next, by equation \eqref{eq:Red2PG3}, we have
\begin{equation*}
  \begin{aligned}
   \delta z_n & = \left(1 - \frac{1 - \gamma}{2} \right) z_n \\
                            & = \min_{i \in [a_n-1]} \left\{ r_i^n + \delta \sum_{k = 1}^{n-1} \frac{\gamma}{\delta} p_{ik}^n z_k + \delta \left(\frac{1}{\delta} -1\right) z_n + \delta \frac{\gamma}{\delta} p^n_{in} z_{n+1} \right\} - \frac{1 - \gamma}{2} z_n \\
														& = \min_{i \in [a_n-1]} \left\{ r_i^n + \gamma \sum_{k = 1}^{n-1} p_{ik}^n z_k + \frac{1-\gamma}{2} z_n + \gamma p^n_{in} z_{n+1} \right\} - \frac{1 - \gamma}{2} z_n \\
														& = \min_{i \in [a_n-1]} \left\{ r_i^n + \gamma \sum_{k = 1}^n p_{ik}^n v_k \right\}.
  \end{aligned}
\end{equation*}
By combining it with equation \eqref{eq:Red2PG4}, we conclude that
 \begin{equation} \label{eq:final1}
  \begin{aligned}
   v_n & = z_{n+1} \\
	     & \leq \delta z_n \\
			 & = \min_{i \in [a_n-1]} \left\{ r_i^n + \gamma \sum_{k = 1}^n p_{ik}^n v_k \right\}. \\
	\end{aligned}
 \end{equation}
On the other hand, we have
\begin{equation} \label{eq:final2}
  \begin{aligned}
   2 v_n & = 2 z_{n+1} \\
	     & \leq 2 \left( \frac{r_{a_n}^n}{2} + \delta \sum_{k = 1}^{n-1} \frac{\gamma}{2 \delta} p_{a_n k}^n z_k + \delta \left( \frac{\gamma}{2 \delta} p^n_{a_n n} + \frac{1}{2 \delta} \right) z_{n+1} \right) \\
			 & = r_{a_n}^n + \gamma \sum_{k = 1}^{n-1} p_{a_n k}^n z_k + \left( \gamma p^n_{a_n n} + 1 \right) z_{n+1} \\
			 & = r_{a_n}^n + \gamma \sum_{k = 1}^n p_{a_n k}^n v_k + v_n.
	\end{aligned}
 \end{equation}
Thus
$$v_n \leq r_{a_n}^n + \gamma \sum_{k = 1}^{n} p_{a_n k}^n v_k.$$
It actually holds that
\begin{equation*}
  \begin{aligned}
   v_n =  \min_{i \in [a_n]} \left\{ r_i^j + \gamma \sum_{k = 1}^n p_{ik}^n v_k \right\}
	\end{aligned}
 \end{equation*}
because either the inequality in \eqref{eq:final1} or the inequality in \eqref{eq:final2} holds with equality. Hence, equation \eqref{eq:2PG2} is likewise satisfied by $v$.
\end{proof}

Finally, it is straightforward to verify that the equations \eqref{eq:Red2PlayerGame} describe an optimality criterion of a game. The transition probability matrix is
$$P':=\begin{bmatrix} \frac{\gamma}{2 \delta} p^1_{\cdot 1} + \frac{1}{2 \delta} \mathbf{1} & \frac{\gamma}{2 \delta} p^1_{\cdot 2} & \frac{\gamma}{2 \delta} p^1_{\cdot 3} & \cdots & \frac{\gamma}{2 \delta} p^1_{\cdot (n-1)} & \mathbf{0} & \frac{\gamma}{2 \delta} p^1_{\cdot n} \\
  \frac{\gamma}{2 \delta} p^2_{\cdot 1} & \frac{\gamma}{2 \delta} p^2_{\cdot 2} + \frac{1}{2 \delta} \mathbf{1} & \frac{\gamma}{2 \delta} p^2_{\cdot 3} & \cdots & \frac{\gamma}{2 \delta} p^2_{\cdot (n-1)} & \mathbf{0} & \frac{\gamma}{2 \delta} p^2_{\cdot n} \\
	\vdots & \vdots & \ddots & \cdots & \vdots & \vdots & \vdots \\
	\frac{\gamma}{2 \delta} p^{n-1}_{\cdot 1} & \frac{\gamma}{2 \delta} p^{n-1}_{\cdot 2} & \frac{\gamma}{2 \delta} p^{n-1}_{\cdot 3} & \cdots & \frac{\gamma}{2 \delta} p^{n-1}_{\cdot (n-1)} + \frac{1}{2 \delta} \mathbf{1} & \mathbf{0} & \frac{\gamma}{2 \delta} p^{n-1}_{\cdot n} \\
	\frac{\gamma}{\delta} p^{n}_{[a_n-1] 1} & \frac{\gamma}{\delta} p^{n}_{[a_n-1] 2} & \frac{\gamma}{\delta} p^{n}_{[a_n-1] 3} & \cdots & \frac{\gamma}{\delta} p^{n}_{[a_n-1] (n-1)} & \left(\frac{1}{\delta} - 1\right) \mathbf{1} & \frac{\gamma}{ \delta} p^{n}_{[a_n-1] n} \\
	\frac{\gamma}{2 \delta} p^n_{a_n 1} & \frac{\gamma}{2 \delta} p^n_{a_n 2} & \frac{\gamma}{2 \delta} p^n_{a_n 3} & \cdots & \frac{\gamma}{2 \delta} p^n_{a_n (n-1)} & \mathbf{0} & \frac{\gamma}{2 \delta} p^n_{a_n n} + \frac{1}{2 \delta} \\
	0 & 0 & 0 & \cdots & 0 & 1 & 0
 \end{bmatrix}.$$

Obviously $P' \geq \mathbf{0}$, and since $\frac{\gamma}{2 \delta} + \frac{1}{2 \delta}= 1 \text{ and } \frac{\gamma}{\delta} + (\frac{1}{\delta} - 1) = 1$,
we have $P' \mathbf{1} = \mathbf{1}$. Hence, matrix $P'$ is rowstochastic. The reward vector is
$$r':=\begin{bmatrix}
      \frac{1}{2} r^1 \\
      \frac{1}{2} r^2 \\
			\vdots \\
			\frac{1}{2} r^{n-1} \\
			r^n_{[a_n-1]} \\
			\frac{1}{2} r^n_{a_n} \\
			0
\end{bmatrix}.$$

This completes an intermediate step of the reduction. The proof for the case $n \in S_{\textsc{max}}$ is analogous and left to the reader.

Through iteratively applying the above reduction scheme, two-player stochastic games can be reduced two games where every state has at most two actions. All in all, the number of states increases linearly in the total number of actions $m$. The discount factor increases in every reduction step by $(1-\gamma)\slash 2$. The game values to the original game are obtained through retracing the reduction steps backwards and using Lemma~\ref{lem:Red2PGGameVal}. In summary, the main result is the following.

\begin{thm}
 Two-player stochastic games reduce to their binary counterparts in strongly polynomial time.
\end{thm}

Since binary stochastic games admit formulations as ordinary P-LCPs, general two-player stochastic games can be formulated as ordinary P-LCPs that again represent games.




\bigskip

\textbf{Acknowledgment.} This work is supported by the JST, ERATO Large Graph Project. We would like to thank Jan Foniok, Komei Fukuda, Naonori Kakimura, and Hanna Sumita for various contributions.

\bibliography{lincomp}

\end{document}